\documentclass[12pt, draftclsnofoot, onecolumn]{IEEEtran}
 \usepackage{amsmath,amssymb}
 \usepackage{subfigure}
 \usepackage{graphicx,graphics,color,psfrag}
 \usepackage{cite,balance}
 \usepackage{caption}
 \allowdisplaybreaks
 \usepackage{accents}
 \usepackage{amsthm}
 \usepackage{bm}
 \usepackage[english]{babel}
 \usepackage{multirow}
 \usepackage{enumerate}
 \usepackage{cases}
 \usepackage{stfloats}
 \usepackage{dsfont}
 \usepackage{color,soul}
 \usepackage{amsfonts}
 \usepackage{cite,graphicx,amsmath,amssymb}
 \usepackage{subfigure}
 \usepackage{fancyhdr}
 \usepackage{hhline}
 \usepackage{graphicx,graphics}
 \usepackage{array,color}
 \usepackage{amsmath}
 \usepackage{booktabs}
 \usepackage{algorithmic,algorithm}
 
\include{header}

\begin{document}
\setlength{\topskip}{-3pt}

\newtheorem{lemma}{Lemma}
\newtheorem{proposition}{Proposition}
\newtheorem{remark}{Remark}

\title{\huge Joint Sensing and Communication-Rate Control for Energy Efficient Mobile Crowd Sensing }
\author{Ziqin Zhou, \emph{Student Member, IEEE}, Xiaoyang Li, \emph{Member, IEEE}, \\ Changsheng You, \emph{Member, IEEE}, Kaibing Huang, \emph{Fellow, IEEE}, and Yi Gong, \emph{Senior Member, IEEE}
\thanks{Ziqin Zhou, Xiaoyang Li, Changsheng You and Yi Gong are with the Department of Electrical and Electronic Engineering (EEE), Southern University of Science and Technology (SUSTech), Shenzhen, China. Kaibin Huang is with the Department of EEE, The University of Hong Kong (HKU), Hong Kong. Corresponding author: Xiaoyang Li (e-mail: lixy@sustech.edu.cn)}
}
\maketitle

\vspace{-15mm}
\begin{abstract}
Driven by the rapid growth of Internet of Things applications, tremendous data need to be collected by sensors and uploaded to the servers for further process. As a promising solution, mobile crowd sensing enables controllable sensing and transmission processes for multiple types of data in a single device. In this paper, a typical user is considered that is required to sense and transmit data to a server, while it is assumed to remain busy and incapable of sensing data during an interval.  An optimization problem is formulated to minimize the energy consumption of data sensing and transmission by controlling the sensing and transmission rates over time, subject to the constraints on the sensing data sizes, transmission data sizes, data casualty, and sensing busy time. This problem is highly challenging, due to the coupling between the rates as well as the existence of the busy time. To deal with this problem, we first show that it can be equivalently decomposed into two subproblems, corresponding to a search for the amount of data size that needs to be sensed before the busy time (referred to as the height), as well as the sensing and transmission rate control given the height. Next, we show that the latter problem can be efficiently solved by using the classical string-pulling method, while an efficient algorithm is proposed to progressively find the optimal height without the exhaustive search. Moreover, the solution approach is extended to a more complex scenario where there is a finite-size buffer at the server for receiving data. Last, simulations are conducted to evaluate the performance of the proposed design.
\end{abstract}


\section{Introduction}
The development of \emph{Internet of Things} (IoT) applications for Smart Cities has enabled a series of services such as auto-driving, pollution assessment, temperature measurement, and public-safety surveillance \cite{chiang2016fog}. Nonetheless, the wide range of services require tremendous data with a variety of types, which need to be collected by the sensors and sent to the servers \cite{ma2014opportunities}. In the conventional \emph{wireless sensor networks} (WSN), it is hard for a sensor to collect diverse types of data due to its limited sensing coverage and scalability \cite{akyildiz2002wireless}. To overcome such bottleneck, \emph{mobile crowd sensing} (MCS) has been proposed by leveraging handheld and wearable IoT devices equipped with multiple sensing modules to collect different types of data \cite{ganti2011mobile}. 

Despite the multiple types of collected data at the sensors, delivering them to the server simultaneously will result in a heavy communication burden if not impossible. Fortunately, the different task requirements with respect to both data size and delay tolerance can be exploited for efficient radio resource management, which is known as the \emph{differentiated radio resource management} (DRRM) \cite{carpenter2002differentiated}. As an intelligent communication method, DRRM integrates different task requirements (such as delay tolerance) and wireless link conditions (such as channel state) to design device access, resource allocation and interference coordination \cite{zhang2014providing}. Different from the traditional communication schemes, DRRM not only focuses on spectrum efficiency, but also on the \emph{quality of service} (QoS) for the sensing tasks \cite{tao2008resource}. For instance, auto-driving needs \emph{ultra reliable low latency communication} (URLLC) of sensed data to ensure safety \cite{popovski20185g}. In contrast, pollution assessment requires \emph{enhanced mobile broadband} (eMBB) for carrying massive sensed data, but is less sensitive to transmission delay \cite{alsenwi2019embb}. In essence, DRRM is an integration between the application layer and the physical layer, i.e., the information of the application layer is used to guide the resource allocation in the physical layer. Based on DRRM, the controllable data sensing and transmission processes of IoT devices can be utilized to improve the efficiency of MCS \cite{li2018wirelessly}.

The investigation of DRRM can be traced back to the earlier works on the long-term rate control, where multiple tasks exist in the same server with different requirements of data size and delay tolerance \cite{prabhakar2001energy,zafer2005calculus,zafer2008optimal,zafer2009calculus,wang2013energy}. It was found in \cite{prabhakar2001energy} that varying packet transmission time can reduce the energy consumption for transmitting given a fixed amount of data. Inspired by such findings, the optimal rate-control policy was proved to have the \emph{string-pulling} (SP) structure for minimizing the total transmission energy while satisfying the QoS constraints of tasks \cite{zafer2005calculus}. In essence, the optimality of SP structure is due to the convexity of the objective function in the optimization problem. The algorithm for deriving the SP structure was further extended to account for the data transmission cases with bursty data arrival \cite{zafer2008optimal}, finite data buffer \cite{zafer2009calculus}, and time-varying channels \cite{wang2013energy}, respectively.

The optimality of rate control with SP structure was further extended to DRRM in other scenarios including \emph{energy-harvesting} (EH) \cite{yang2011optimal,tutuncuoglu2012optimum,devillers2012general,gurakan2013energy,ozel2013optimal,wang2014optimal,ulukus2015energy}, \emph{channel estimation} (CE) \cite{luo2012training}, \emph{relaying} \cite{huang2012throughput}, \emph{caching} \cite{gregori2016wireless}, \emph{mobile edge computing} (MEC) \cite{you2018exploiting}, and \emph{edge learning} (EL) \cite{li2021data}. In EH systems, the power control policy with SP structure was proved to be optimal for transmission delay minimization given the profiles of energy arrivals \cite{yang2011optimal}, battery capacity \cite{tutuncuoglu2012optimum}, battery leakage \cite{devillers2012general}, energy cooperation \cite{gurakan2013energy}, channel states \cite{ozel2013optimal}, circuit power consumption \cite{wang2014optimal}, and user states \cite{ulukus2015energy}. As for CE, the joint training period and power control based on SP structure was proved to be optimal for minimizing the channel estimation error \cite{luo2012training}. Taking relay into consideration, the SP structure based joint source and relay power allocation over time was designed for maximizing the throughput \cite{huang2012throughput}. As for caching, the joint design of the transmission and caching policies based on SP structure was carried out for minimizing the traffic and energy costs over the backhaul links \cite{gregori2016wireless}. In MEC systems, the offloading data size with the SP structure over the whole computing duration was proved to be optimal for energy consumption minimization given the \emph{central processing unit} (CPU) state information \cite{you2018exploiting}. As for EL, the jointly data partition and transmission rate design based on SP structure was proposed for minimizing the transmission energy consumption as well as the classification errors of multiple learning tasks including support vector machines and convolutional neural networks \cite{li2021data}. Despite the rich literatures on SP structure based DRRM, all of them only focus on data transmission process without taking the data sensing process into consideration.

By relaxing the simple assumption of one-shot or bursty data arrivals in literatures \cite{you2018exploiting,li2021data}, the data sensing process of MCS device is controllable \cite{lane2010survey}. Therefore, the sensing rate can be optimized together with the transmission rate to minimize the total energy consumption for data sensing and transmission. Moreover, as the sensing modules of MCS device might be occupied by other applications (e.g., the camera is occupied when the device holder is in video call), the sensing process need to be suspended in such duration, which is known as the \emph{busy time interval} \cite{khan2012mobile}. During such an interval, the sensing rate has to be zero, while the transmission module can work as normal. However, as the data need to be sensed before transmission, the sensing process and transmission process naturally interact with each other. The interdependence of sensing and transmission processes together with the existence of busy time result in coupling optimization variables and nontrivial problem. By exploiting the SP structure as well as taking geometrical analysis in this paper, the optimal sensing and transmission rate controls are derived for the scenarios with finite or infinite data buffer capacity.

The main contributions of this work are summarized below.

\begin{itemize}
\item \emph{Optimal Sensing and Transmission Rate Control given Fixed Height:} Given the requirements of tasks w.r.t. data size and delay tolerance, the data sensing and transmission rate are jointly optimized to minimize the total energy consumption of sensing and transmission. Due to the fact the the data need to be sensed before transmission, the sensing and transmission rates are coupled together, which makes the optimization problem non-trivial. The existence of busy time interval in practice further complicates the problem. To obtain the tractable solution, a vital concept namely height is introduced in this paper, which specifies the amount of data that needs to be sensed before the busy time. Given the fixed height, the optimal sensing and transmission rates controlling policies can be obtained by using the SP method. 
     
\item \emph{Optimal Design with Infinite Server Data Buffer Capacity:} For a server with infinite data buffer capacity, the original rates optimization problem can be converted to a search for the optimal height. Based on the property of the objective function, the upper and lower bounds of the area where the optimal height lies in are derived. The whole searching area is further divided into a series of sub-areas due to the different expressions of the objective function with respect to different heights.. Finally, the optimal height in each sub-area is obtained based on the convexity of the objective function and the global optimal height is further determined by comparing the local optimums.

\item \emph{Optimal Design with Finite Server Data Buffer Capacity:} 
The versatility of the above solving approaches is further demonstrated by an extension to the case with finite data buffer capacity of the server. Specifically, the transmitted data size is upper bounded by the data buffer capacity and thus the algorithm is adjusted correspondingly, while the SP structure also holds in the optimal design for this case.  
\end{itemize}
    
The rest of this paper is organized as follows. The system model and problem formulation are described in Section II. The optimal design of sensing and transmission rates control in the case with infinite server data buffer capacity is presented in Section III, which is further extended to account for the case with finite server data buffer capacity in Sections IV. The experimental results are presented in Section VI. The conclusions are drawn in Section VII.

\section{System Model and Problem Formulation}
\subsection{System Model}
Consider the system as shown in Fig.~\ref{FigSys}, a server has $N$ tasks to be executed at time instants $\{t_1,t_2,...,t_N\}$. Completing the $n$-th task (e.g., learning model training \cite{li2021data}) requires $D_n$ amount of data sensed and wirelessly transmitted by mobile devices before instant $t_n$. Both the sensing and transmission rates are varying from time to time, denoted by $s(t)$ and $r(t)$ respectively. Let $C$ denote the number of \emph{central processing unit} (CPU) cycles for sensing one bit of data, the CPU cycle frequency can be determined by $f(t) = s(t)C$. Following the models in \cite{chandrakasan1992low,you2018exploiting,you2016energy}, under the assumption of a low CPU voltage, the sensing power consumption $P_s(t) = \alpha f^2(t) = \alpha C^2 s^2(t)$, where $\alpha$ is a constant determined by the circuits. According to the Shannon capacity \cite{goldsmith2005wireless}, the transmission power consumption $P_t(t) = \frac{\sigma^2}{g}(e^{r(t)/B}-1)$ with $\sigma^2$, $g$, and $B$ denoting the noise power, effective channel power gain, and spectrum bandwidth respectively. The objective is to minimize the energy consumption for data sensing and transmission, i.e,
\begin{equation}\label{Eq:minobjective}
\min_{s(t),r(t)} \int_{t=0}^{t_N} \left [\alpha C^2 s^2(t)+\frac{\sigma^2}{g}(e^{r(t)/B}-1)\right]dt.
\end{equation}
To guarantee the execution of tasks, both the sizes of the sensed and transmitted data should be no less than $D_n$ before instant $t_n$, which are known as the data requirement constraints:
\begin{equation}
\text{(Data sensing requirement)} \int_{t=0} ^{t_j}s(t) \geq \sum_{n=1}^j D_n, j = 1,2,...N, \label{Eq:sensearrival}
\end{equation}
\begin{equation}
\text{(Data transmission requirement)} \int_{t=0}^{t_j}r(t) \geq \sum_{n=1}^j D_n, j =1,2,...,N. \label{Eq:transarrival}
\end{equation}
If the amount of data that can be sensed or transmitted in the $n$-th epoch is larger than $D_n$, such epoch can be used to sense or transmit extra data required by the subsequent tasks. The extra sensed data and transmitted data are stored in the buffers of the mobile device and the server respectively, whose capacities are both assumed to be infinite. Since data needs to be sensed before transmission, the amount of transmitted data should be no longer than that of sensed data through the whole duration, which is known as the casualty constraint and expressed as
\begin{equation}\label{Eq:casualty}
\int_{t=0}^{\tau}s(t)dt \geq \int_{t=0}^{\tau}r(t)dt, \forall \tau \in [0,t_N].
\end{equation}

Despite data sensing, the mobile devices may have other tasks. Therefore, the CPU of mobile device might be busy and no data can be sensed during some time intervals, e.g., $[b_1,b_2]$ as shown in Fig.~\ref{FigSys}. The time sequence $\{t'_i\} = \{t_1,...,t_{i},b_1,t_{i+1},...,t_{j},b_2,t_{j+1},...,t_N\}$ are divided into $N' = N+2$ epochs with the length $T'_i = t'_i - t'_{i-1}$ for $i = 1,...,N'$ and $t'_0 = 0$. 

\begin{figure}[t]
\centering
\includegraphics[scale=0.5]{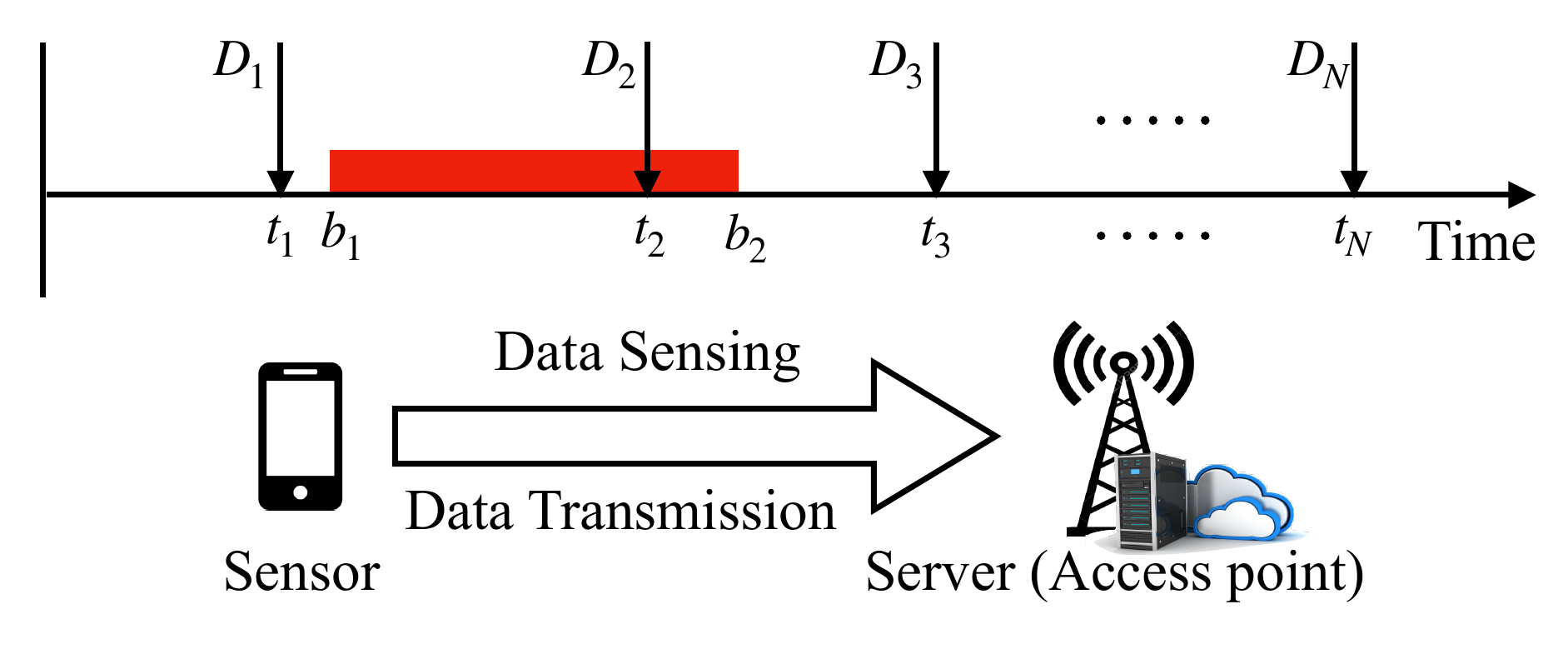}
\caption{MCS system with multiple tasks}
\label{FigSys}
\end{figure}


\subsection{Problem Formulation}
Based on the above discussion, the corresponding optimization problem can be formulated as 
\begin{subequations}
\begin{align}
\min_{s(t),r(t)} & \int_{t=0}^{t_N} \left [\alpha C^2 s^2(t)+\frac{\sigma^2}{g}(e^{r(t)/B}-1)\right]dt \label{Eq:P1a}\\ 
\text{(P1)} \qquad \text{s.t.} \quad
& \int_{t=0} ^{t_j} s(t) \geq \sum_{n=1}^j D_n, j = 1,...N ,\label{Eq:P1b}\\ 
& \int_{t=0}^{t_j}r(t) \geq \sum_{n=1}^jD_n, j =1,...,N,\label{Eq:P1c}\\ 
&\int_{t=0}^{\tau}s(t)dt \geq \int_{t=0}^{\tau}r(t)dt, \forall \tau \in [0,t_N] ,\label{Eq:P1d}\\ 
&s(t) = 0 ,\forall t \in [b_1,b_2].\label{Eq:P1e}
\end{align}
\end{subequations}
It should be noted that $\int _{t=0}^{t_N}s(t) dt = \int _{t=0}^{t_N}r(t) dt = \sum_{n=1}^{N}D_n$ must hold when the optimum is achieved, otherwise one can always decrease $s(t)$ or $r(t)$ without conflicting other constraints, and thus reduce the power consumption. The constraints in \eqref{Eq:P1b} and \eqref{Eq:P1c} guarantee the sensing and transmission of the required amount of data. The constraints in \eqref{Eq:P1d} indicate that the data should be sensed before transmitted. The constraints in \eqref{Eq:P1e} specify the duration when no data can be sensed.

\section{Joint Optimization of Sensing and Transmission Rates}
In this section, the original problem is first simplified without loss of optimality by converting the continuous variables $\{r(t)\}, \{s(t)\}$  into discrete variables $\{r_i\},\{s_i\}$ based on the structure of rate optimal control. To solve the simplified problem, an algorithm namely \emph{height search} based on the SP structure is proposed.

\subsection{Structure of Optimal Rate Control}
For the energy minimization problem, the following lemma shows that the constant-rate sensing and transmission within each epoch is optimal.

\begin{lemma}[Optimality of Constant Rate]\label{Lemma:Constant}\emph{In each epoch, the constant-rates sensing and transmission are optimal.
}
\end{lemma}

\begin{proof}
Assume that there are two sensing rates before and after instant $t'_i \in [t'_a,t'_b)$, denoted as $s_i$ and $s_{i+1}$ respectively. The sensing energy consumption is $E_s = \alpha C^2 s_i^2 (t'_i-t'_a) + \alpha C^2 s_{i+1}^2 (t'_b-t'_i)$. Let $s' = \frac{s_i (t'_i - t'_a) + s_{i+1} (t'_b - t'_i)}{t'_b - t'_a}$ denote the new sensing rate over $[t'_a,t'_b)$, the sensing power becomes $P_s' = \alpha C^2 \left(\frac{s_i (t'_i - t'_a) + s_{i+1} (t'_b - t'_i)}{t'_b - t'_a}\right)^2$. Due to the convexity, $P_s' \leq \alpha C^2 s_i^2 \frac{t'_i - t'_a}{t'_b - t'_a} + \alpha C^2 s_{i+1}^2 \frac{t'_b - t'_i}{t'_b - t'_a}$. The sensing energy consumption over this duration is $E_s' = \alpha C^2 (t'_b - t'_a) \left(\frac{s_i (t'_i - t'_a) + s_{i+1} (t'_b - t'_i)}{t'_b - t'_a}\right)^2 \leq E_s$. Therefore, the energy consumption under the new policy is less than that under the original policy for sensing the same amount of data in this epoch, and thus the original policy cannot be optimal. Following the similar approach, the transmission rate has the same property.
\end{proof}

Based on Lemma 1, the optimal sensing and transmission rates $s(t)$ and $r(t)$ might only change at $\{t'_i\}$. The corresponding rates in the $i$-th epoch are denoted by $s_i$ and $r_i$, respectively. As there is no task to be executed at instants $b_1$ and $b_2$, $D'_i = 0$ when $t'_i = b_1$ or $t'_i = b_2$, while $D'_i = D_n$ when $t'_i = t_n$. Therefore, the original problem can be converted to: 
\begin{subequations}
\begin{align}
\min_{\{r_i \geq 0\},\{s_i \geq 0\}}  ~
&  \sum_{i=1}^{N'} \left [\alpha C^2 s_i^2+\frac{\sigma^2}{g}(e^{r_i/B}-1)\right]T'_i \label{Eq:P2a}\\
\text{(P2)} \qquad \text{s.t.} \qquad
& \sum_{i=1}^j s_i T'_i \geq \sum_{i=1}^j D'_i, j = 1,...N',\label{Eq:P2b}\\
& \sum_{i=1}^j r_i T'_i\geq \sum_{i=1}^j D'_i, j = 1,...N',\label{Eq:P2c}\\
& \sum_{i=1}^j s_i T'_i \geq \sum_{i=1}^j r_i T'_i, j = 1,...N',\label{Eq:P2d}\\
& s_i = 0, \forall t'_i \in [b_1,b_2]. \label{Eq:P2e}
\end{align}
\end{subequations}
It should be noted that the coupling rates in \eqref{Eq:P2d} and the existence of busy time in \eqref{Eq:P2d} make (P2) non-trivial. A feasible solution is illustrated in Fig.~\ref{FigHeight}. To derive the optimal solution, an useful concept namely \emph{height $h$ is introduced, which denotes the amount of data that needs to be sensed before the busy time interval $b_1$}. 

\begin{figure}[t]
\centering
\includegraphics[scale=0.5]{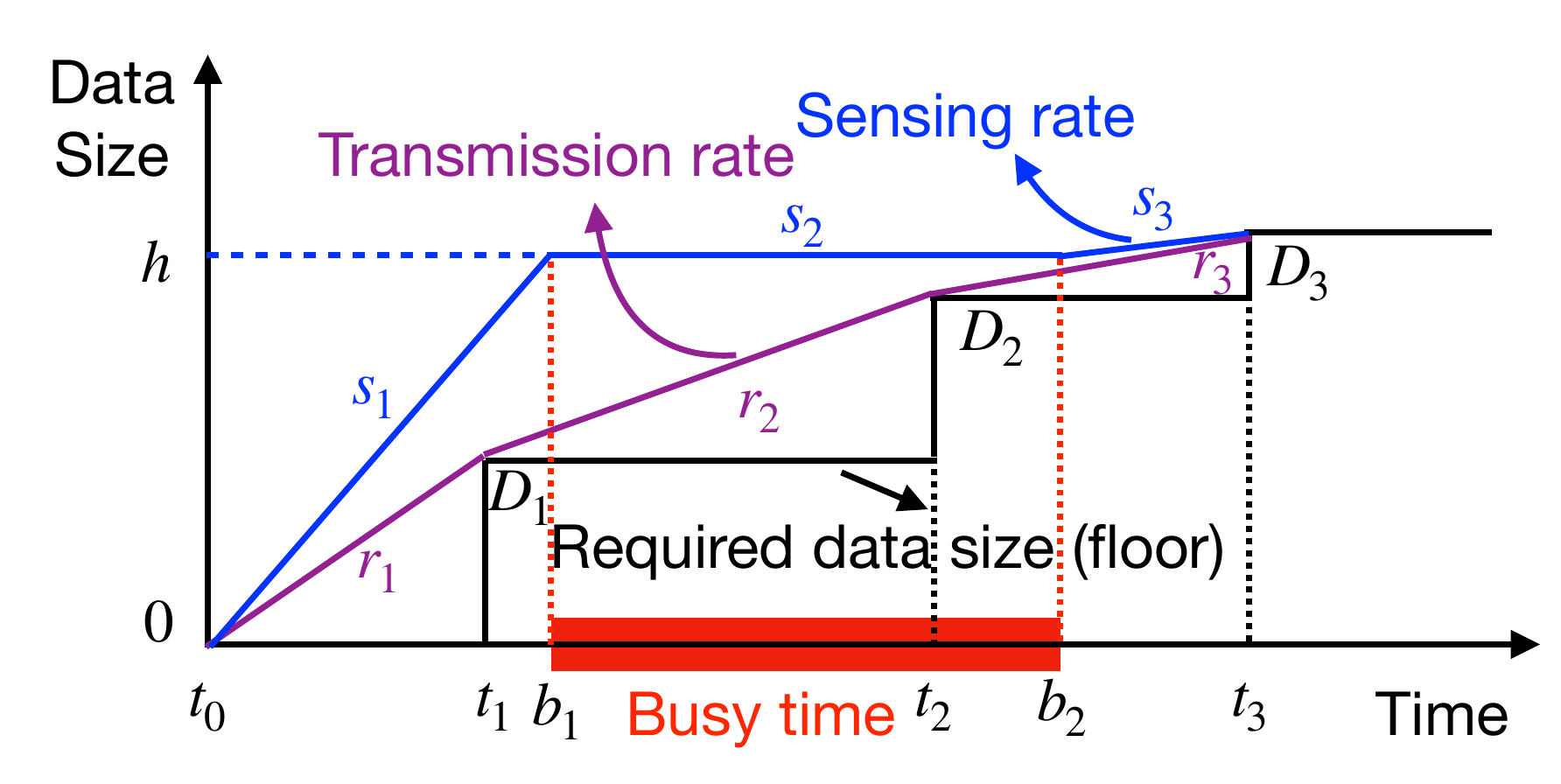}
\caption{A feasible rate control scheme}
\label{FigHeight}
\end{figure}


\subsection{Optimal Rates based on Given Height}
Given a certain height $h$, the optimization of sensing rates before the busy time interval aims at constructing the shortest path from a starting point (e.g., $(t_0,0)$ in Fig.~\ref{FigHeight}) to an ending point (e.g., $(b_1,h)$ in Fig.~\ref{FigHeight}) above the \emph{floor determined by the required data size $\{D_n\}$}, which can be performed by \emph{pulling a stretched string from the starting point to the ending point above the floor} and known as SP. The SP structure also holds for optimizing the sensing rates after the busy time interval by pulling a stretched string from $(b_2,h)$ to $(t_N,\sum_{n=1}^N D_N)$ above the floor determined by $\{D_n\}$. The corresponding optimized rates have the following property:




\begin{lemma}[Optimal Rate Control above the Floor]\label{Lemma:RateUn}\emph{The optimal rate above the floor is non-increasing along the epochs. Moreover, the rate might decrease only when the data requirement constraint is active.
}
\end{lemma}

\begin{proof}
Due to the convexity of $P_s(t)$, if two feasible rates in consecutive epochs $i$ and $i+1$ satisfy $s_i < s_{i+1}$, one can always find a feasible rate $s' = \frac{s_i T'_i + s_{i+1} T'_{i+1}}{T'_i + T'_{i+1}}$ with less energy consumption. Moreover, if there exists a gap between the sensed and required data size, there exists room for equalizing the rates in two epochs to reduce the energy consumption. Such property also holds for data transmission with the specified proof in Section III-B of \cite{li2021data}.
\end{proof}

Based on the non-increasing property of the rate given in Lemma \ref{Lemma:RateUn}, the SP algorithm in \cite{li2021data} is applied to obtain the optimal sensing rates in each epoch by iteratively finding the rate curve with the largest slope, which is summarized in Algorithm \ref{A1}. As for transmission rate optimization, \emph{the feasible region is a tunnel with the floor determined by the required data size $\{D_n\}$, and the ceiling determined by the height $h$ during the busy time interval} due to the casualty constraint \eqref{Eq:casualty}. The optimized rates based on SP in such tunnel have the following property:

\begin{algorithm}[tt]
\renewcommand{\algorithmicrequire}{\textbf{Input:}}
\renewcommand{\algorithmicensure}{\textbf{Output:}}
\caption{SP Algorithm for Optimal Rate Control above the Floor.}
\label{A1}
\begin{algorithmic}[1]
\REQUIRE required data sizes $\{D_n\}$ at instants $\{t_n\}$ for $N$ tasks.
\ENSURE the optimal rates $\{r_i^*\}$ and durations $\{T_i^*\}$.
\STATE Initialize $j_0 = 0$, $i = 0$.
\STATE \textbf{while} $j_i < N$
\STATE \qquad Update $i = i+1$.
\STATE \qquad Calculate $j_i = \arg\max_{j:j_{i-1}<j \leq N}\left\{\frac{\sum_{n=j_{i-1}+1}^{j}D_n}{t_j-t_{j_{i-1}}}\right\}$. 
\STATE \qquad Calculate $r_i^*= \left\{\frac{\sum_{n=j_{i-1}+1}^{j_i}D_n}{t_{j_i}-t_{j_{i-1}}}\right\}$.
\STATE \qquad Calculate $T_i^* = t_{j_i}-t_{j_{i-1}}$.
\STATE \textbf{End while}
\STATE \textbf{Return} the optimal sensing rates $\{r_i^*\}$ and durations $\{T_i^*\}$.
\end{algorithmic}
\end{algorithm}

\begin{lemma}[Optimal Rate Control in the Tunnel]\label{Lemma:RateLimited}\emph{The optimal rate in the tunnel might only decrease when the data requirement constraint is active, and might only increase when the casualty constraint is active.
}
\end{lemma}

\begin{proof}
Suppose that the rate changes at arbitrary time $t$, so that $r(t^-) \neq r(t^+)$. Let $r' = \frac{r(t^-)+r(t^+)}{2}$ be the constant rate in $[t-\tau,t+\tau]$. If $r(t^-) < r(t^+)$, then $r'$ is feasible only when the casualty constraint is inactive. If $r(t^-) > r(t^+)$, then $r'$ is feasible only when the data requirement constraint is inactive. According to Lemma \ref{Lemma:Constant}, the application of $r'$ can reduce the energy consumption and thus the original rates are not optimal. The specified proof can be found in Section IV-A of \cite{li2021data}.
\end{proof}

Based on Lemma \ref{Lemma:RateLimited}, the SP algorithm in \cite{li2021data} is applied to obtain the optimal transmission rates in each epoch by iteratively finding the feasible region and the corresponding constant rates, which is summarized in Algorithm \ref{A2}. Though the optimal sensing and transmission rates based on a certain height can be obtained by applying Algorithms \ref{A1} and \ref{A2} separately, the casualty constraints in the non-busy time intervals are ignored. Fortunately, it can be proved that the rates determined by Algorithms \ref{A1} and \ref{A2} won't violate the casualty constraints, the whole process for finding the optimal sensing and transmission rates is summarized in Algorithm \ref{A3}. 

\begin{algorithm}[tt]
\renewcommand{\algorithmicrequire}{\textbf{Input:}}
\renewcommand{\algorithmicensure}{\textbf{Output:}}
\caption{SP Algorithm for Optimal Rate Control in the Tunnel.}
\label{A2}
\begin{algorithmic}[1]
\REQUIRE required data sizes $\{D_n\}$ at instants $\{t_n\}$, maximum data sizes $\{A(t_m)\}$ at instants $\{t_m\}$.
\ENSURE the optimal rates $\{r_i^*\}$ and durations $\{T_i^*\}$.
\STATE Initialize $v_b = 0$, $v_1 = 0$, $i = 0$, $n_1 = 0$.
\STATE \textbf{while} $M > 0$
\STATE \qquad Update $i = i+1$.
\STATE \qquad \textbf{for} $m = 1,...,M$
\STATE \qquad \qquad $r_{\text{low}}[m] = \frac{\sum_{n: 0 \leq t_n < t_m} D_n}{t_m}$,
\STATE \qquad \qquad $r_{\text{high}}[m] = \frac{A(t_m)}{t_m}$,
\STATE \qquad \qquad $\bold{r}[m] = [r_{\text{low}}[m],r_{\text{high}}[m]] = \{r|r_{\text{low}}[m] \leq r \leq r_{\text{high}}[m]\}$.
\STATE \qquad \textbf{end for}
\STATE \qquad Update $v_b = \max\left\{v|\bigcap_{m=1}^u\bold{r}[m]\neq \emptyset, m=1,2,...,M\right\}$.
\STATE \qquad \textbf{if} $v_b = M$ 
\STATE \qquad \qquad Update $v_1 = \max\left\{v|r_{\text{low}}[v] \in \bigcap_{j=1}^{v_b}\bold{r}[m]\right\}$, $r_i^* = r_e[v_1]$, $T_i^* = t_{v_1}$.
\STATE \qquad \textbf{else} 
\STATE \qquad \qquad \textbf{if} $\bold{r}[v_b+1]$ falls below $\bigcap_{m=1}^{v_b}\bold{r}[m]$
\STATE \qquad \qquad \qquad Update $v_1 = \max\left\{v|r_{\text{low}}[v] \in \bigcap_{m=1}^{v_b}\bold{r}[m]\right\}$, $r_i^* = r_{\text{low}}[v_1]$ and $T_i^* = t_{v_1}$.
\STATE \qquad \qquad \textbf{else}
\STATE \qquad \qquad \qquad Update $v_1 = \max\left\{v|r_{\text{high}}[v] \in \bigcap_{m=1}^{v_b}\bold{r}[m]\right\}$, $r_i^* = r_{\text{high}}[v_1]$ and $T_i^* = t_{v_1}$.
\STATE \qquad \qquad \textbf{end if}
\STATE \qquad \textbf{end if}
\STATE \qquad Update $M \!=\! M \!-\! v_1$, $t_m \!=\! t_{m+v_1} \!-\! t_{v_1}$, $n_1 \!=\! \max\{n|t_n \!\leq\! t_{v_1}\}$, $t_n \!=\! t_{n+n_1} \!-\! t_{v_1}$, $t_0 \!=\! 0$.
\STATE \qquad Update $D' = r_i^* T_i^* - \sum_{n=0}^{n_1} D_n$, $D_n = D_{n+n_1}$, $D_1 = D_1 - D'$, $A(t_m) = A(t_{m+v_1}) - r_i^* T_i^*$.
\STATE \textbf{end while}
\STATE \textbf{Return} the optimal rates $\{r_i^*\}$ and durations $\{T_i^*\}$.
\end{algorithmic}
\end{algorithm}



\begin{proposition}[Optimal Sensing and Transmission Rates Given Fixed Height] \label{prop:Optfixedh} \emph{The solutions obtained by Algorithm \ref{A3} are the optimal sensing and transmission rates given the fixed height.}
\end{proposition}

\begin{proof}
It can be observed that the only constraint that causes the coupling relationship between sensing and transmission rates is constraint \eqref{Eq:P2d}. Without this constraint, the optimal sensing rate $s^*(t)$ and transmission rate $r^*(t)$ can be derived by algorithms 1 and 2, respectively. Then we only need to prove that the rates derived by algorithms 1 and 2 are feasible w.r.t. the constraint \eqref{Eq:P2d}. Assume that constraint \eqref{Eq:P2d} is violated for the interval $[t_i , t_{i+1}]$, we can decrease the transmission rate $r(t)$ such that $r(t) = s(t)$ in $[t_i , t_{i+1}]$. Since $\int_{t=0}^{t_j} s(t) dt \geq \sum_{n=1}^j D_n$, so does $\int_{t=0}^{t_j} r(t) dt \geq \sum_{n=1}^j D_n$. Therefore, the transmission energy can be reduced without violating any constraints, which contradicts the optimality of $r^*(t)$ and thus impossible.
\end{proof}

\begin{remark}[Optimal Rates without Busy Time]\label{Rem:NoBusy}\emph{When there is no busy time, both the optimal sensing and transmission rates have the SP structures and can be derived by Algorithm~\ref{A1}, which are equal to each other in all time slots, i.e., $s(t) = r(t)~\forall~t \in [0,t_N]$.
}
\end{remark}

\begin{algorithm}[tt]
\renewcommand{\algorithmicrequire}{\textbf{Input:}}
\renewcommand{\algorithmicensure}{\textbf{Output:}}
\caption{Optimal Sensing and Transmission Rates Searching Algorithm given Fixed Height.}
\label{A3}
\begin{algorithmic}[1]
\REQUIRE Required data amounts $\{D_n\}$ at instants $\{t_n\}$ for $N$ tasks, busy time interval $[b_1,b_2]$, and the amount of sensed data (height) $h$ before the busy time.
\ENSURE Optimal transmission rates $\{r_i^*\}$ and durations $\{{T_i^{r}}^*\}$, optimal sensing rates $ \{s_i^*\} $ and durations $\{{T_i^{s}}^*\}$.
\STATE Initialize $n_1 = \max \{n|t_n < b_1\}$, $n_2 = \min\{n|t_n > b_2\}$.
\STATE Set $\{D_n^1\} = \{D_0,D_1,D_2,...,D_{n_1},h\}$ at instants $\{t_n^1\} = \{t_0,t_1,t_2,...,t_{n_1},b_1\}$.
\STATE Set $\{D_n^2 \} = \{h,D_{n_2}-h,D_{n_2+1}-h,...,D_n-h\}$ at instants $\{t_n^2\} = \{b_2,t_{n_2},t_{n_2 + 1},...,t_n\}$.
\STATE Given $\{D_n^1\}$, $\{D_n^2\}$, $\{t_n^1\}$ and $\{t_n^2\}$, find the optimal sensing rates $\{s_i^1\}$, $\{s_i^2\}$ and durations $\{T_i^1\}$, $\{T_i^2\}$ by Algorithm \ref{A1}.
\STATE Get the optimal sensing rates $\{s_i ^*\} = \{\{s_i^1\},0,\{s_i^2\}\}$ and durations $\{{T_i^{s}}^*\} = \{\{T_n ^1\},b_2 - b_1,\{T_n ^2\}\}$.
\STATE Obtain $\{t_m\}$ by sorting the instants $\{t_n\},b_1,b_2$ in increasing order. 
\STATE Set $A(t_m) = h$ if $t_m = b_1$ or $t_m = b_2 $. Otherwise, $A(t_m) = +\infty$.
\STATE Find the optimal transmission rates $\{r_i^*\}$ and durations $\{{T_i^{r}}^*\}$ by Algorithm \ref{A2}.
\STATE \textbf{Return} $\{r_i^*\}, \{T_i ^*\},\{s_i ^*\},\{T_i^{**}\}$
\end{algorithmic}
\end{algorithm}
 
 \subsection{Searching Area for the Optimal Height} 
According to Proposition~\ref{prop:Optfixedh}, the remaining problem is to search for the optimal height. The first part of height searching algorithm is to find a specific area where the optimal height lies in, which depends on the lemmas below. 

\begin{lemma}[Minimum Sensing Energy Consumption]\label{lemma:MinSense} \emph{The sensing energy consumption achieves its minimum when the sensing rates in the epochs adjacent to the busy time interval are equal. }
\end{lemma}

\begin{proof}
Given the height $h$, the sensing energy consumption in the epochs adjacent to the busy time interval will be $\alpha C^2 b_1(\frac{h}{b_1})^2+ \alpha C^2 (T-b_2) (\frac{D-h}{t-b_2})^2 $, which is firstly decreasing and then increasing with increasing $h$. Let the derivative w.r.t $h$ equal to 0, one can get $\frac{h}{b_1} = \frac{D-h}{t-b_2}$. 
\end{proof}

\begin{remark}[Effect of Height on Sensing Rates]\label{Rem:Hight}\emph{It can be observed from the proof of Lemma~\ref{lemma:MinSense} that the sensing rate before the busy time is increasing with the height, while the sensing rate after the busy time is decreasing with the height.
}
\end{remark}

\begin{figure}[t]
\centering
\includegraphics[scale=0.5]{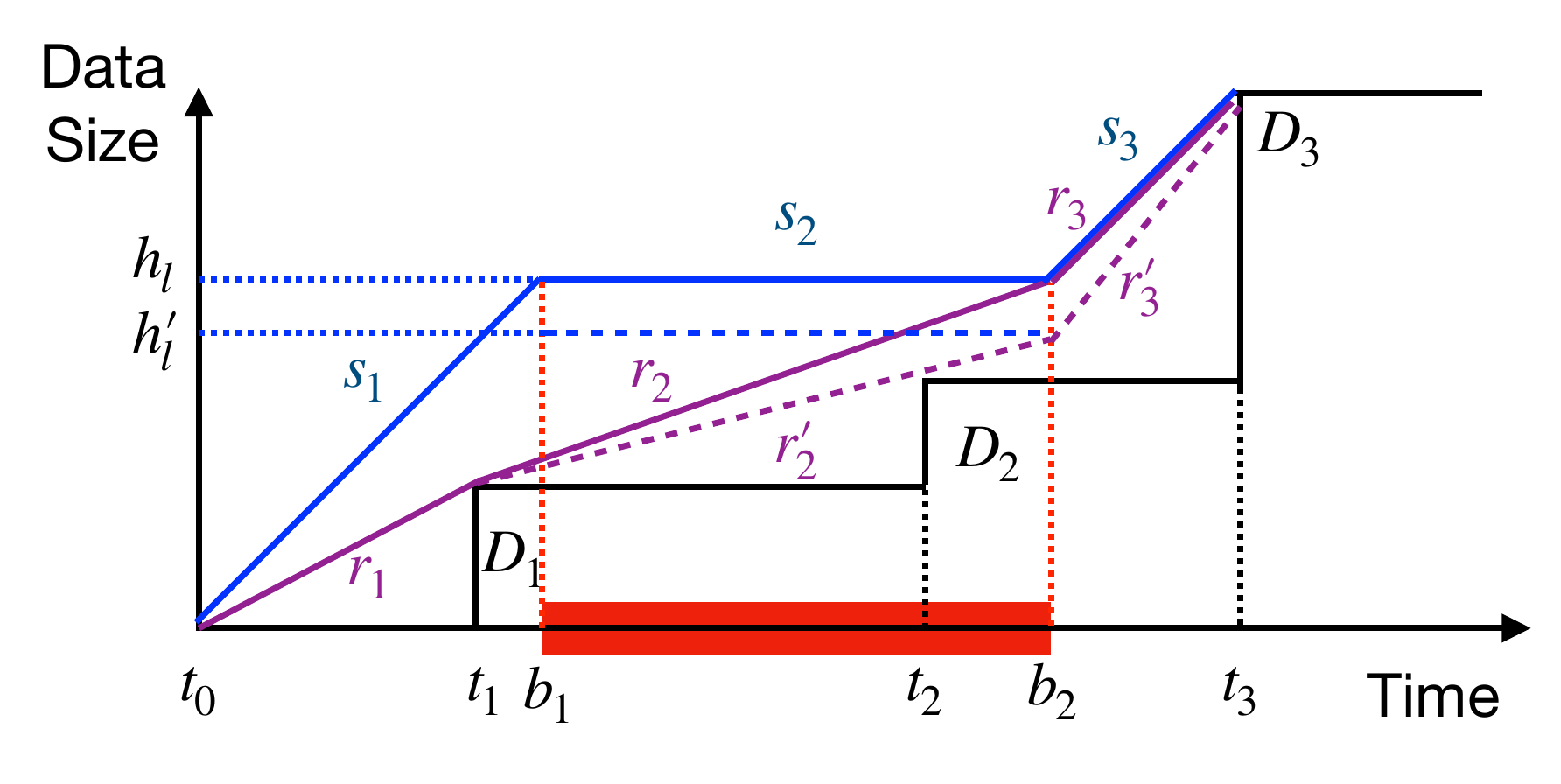}
\caption{The lower bound of searching area for optimal height}
\label{FigLower}
\end{figure}

Based on Lemma \ref{lemma:MinSense}, the lower bound of the searching area for the optimal height is obtained in the following proposition. 

\begin{proposition}[Lower Bound of the Searching Area]\label{prop:LowBound} \emph{The lower bound of the searching area is the height that the sensing rates in the epochs adjacent to the busy duration are equal if it is feasible. Otherwise the lower bound is the required data size by the end of the busy time.}
\end{proposition}

\begin{proof}
According to Lemma \ref{lemma:MinSense}, the sensing energy consumption won't be reduced by decreasing the height. As for transmission, since the height represents the amount of sensed data, decreasing the height will make the casualty constraint more stringent. As shown in Fig.~\ref{FigLower}, suppose that the optimal transmission rates are $\{r_1,r_2,r_3\}$ given the height $h_l$. By decreasing the height to $h_l'$, the transmission rates turns to $\{r_1,r_2',r_3'\}$. Due to the convexity of the Shannon capacity, the transmission energy consumption $E_t(r_2)+E_t(r_3) \leq E_t(r_2')+E_t(r_3')$ as $r_2' < r_2 < r_3 < r_3'$. That is to say, the transmission energy consumption won't be reduced by decreasing the height. Therefore, such height is the lower bound of the searching area if it is feasible. If $h_l$ is not feasible, the lower bound is the required data size by the end of the busy time since it is the lowest achievable height towards the minimum sensing energy consumption.
\end{proof}

On the other hand, the upper bound can be determined by optimizing the transmission and sensing rates sequentially, which is given in the following proposition.

\begin{proposition}[Upper Bound of the Searching Area]\label{prop:UpBound} \emph{The upper bound of the searching area is the height that is derived by solely optimizing the transmission rates without considering sensing.}
\end{proposition}

\begin{proof}
{As shown in Fig.~\ref{FigUpper}, the optimal transmission rates without taking sensing into consideration has the SP structure with the floor determined by the required data size, and thus can be derived by Algorithm~\ref{A1}. The corresponding height at the end of busy time $b_2$ is $h_u$. According to Lemma \ref{Lemma:RateUn}, the optimal transmission rates won't change by increasing $h_u$ and so does the transmission energy consumption. Given the height $h_u$, the optimal sensing rates before and after the busy time interval can also be derived by Algorithm~\ref{A1}. Due to the casualty constraint, one can get $s_1 \geq r_1$. As for the epoch after the busy time interval, since the point $(b_2,h_u)$ lies on the curve of $r_3$, one can get $r_3 = s_3$ according to Algorithm~\ref{A1}. According to Lemma \ref{Lemma:RateLimited}, one can get $s_1 \geq r_1 \geq r_2 \geq r_3 = s_3$. That is to say, the sensing rate before the busy time is no less than that after busy time, which indicates that the height $h_u \geq h_l$. Therefore, the sensing energy consumption will not be reduced by increasing the height $h_u$. Since the total energy consumption will not be reduced by increasing the height $h_u$, such height is the upper bound.}
\end{proof}

\begin{figure}[t]
\centering
\includegraphics[scale=0.5]{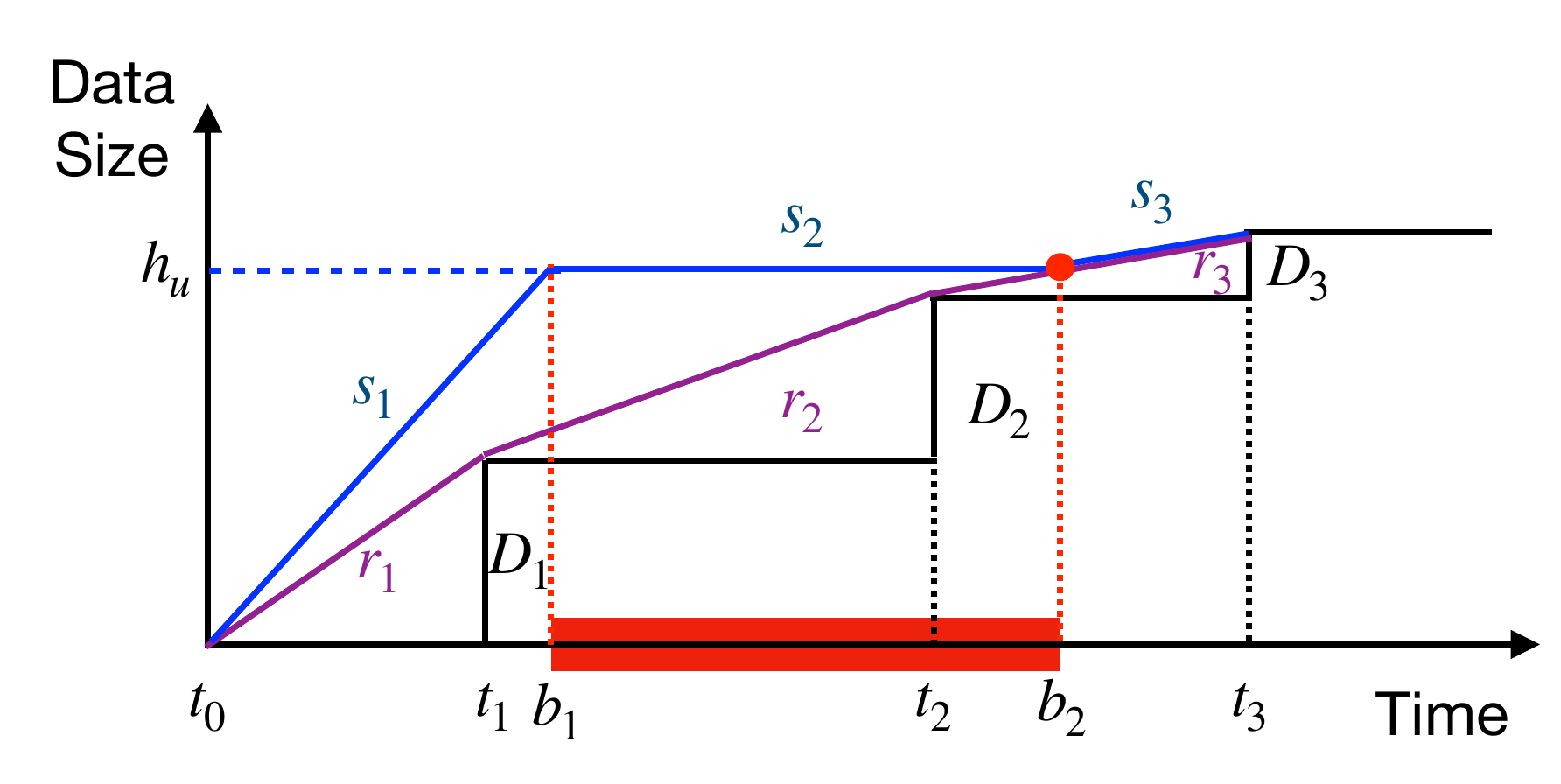}
\caption{The upper bound of searching area for optimal height}
\label{FigUpper}
\end{figure}

Moreover, one can also observe from Fig.~\ref{FigUpper} that the sensing and transmission rate curves are overlapping after the busy time interval. Indeed, such a property is valid accounting for all heights in the searching area as given in the following lemma.

\begin{lemma}[Sensing and Transmission after the Busy Time Interval]\label{lemma:Overlap} \emph{For all heights in the searching area, the sensing and transmission rate curves are overlapped after the busy time interval.}
\end{lemma}
\begin{proof}
The first step is to prove that the amount of sensed data should be equal to that of transmitted data by the end of the busy time interval. Otherwise the data that has been sensed should be more than that has been transmitted due to the casualty constraints. In such case, decreasing the height will reduce the sensing energy consumption and thus contradicts the optimality. Therefore, both the sensing and transmission rate curves has the same start and end points after the busy time interval. According to Remark~\ref{Rem:NoBusy}, these two curves are overlapping.
\end{proof} 

\begin{figure}[t] 
\centering
\includegraphics[scale = 0.5]{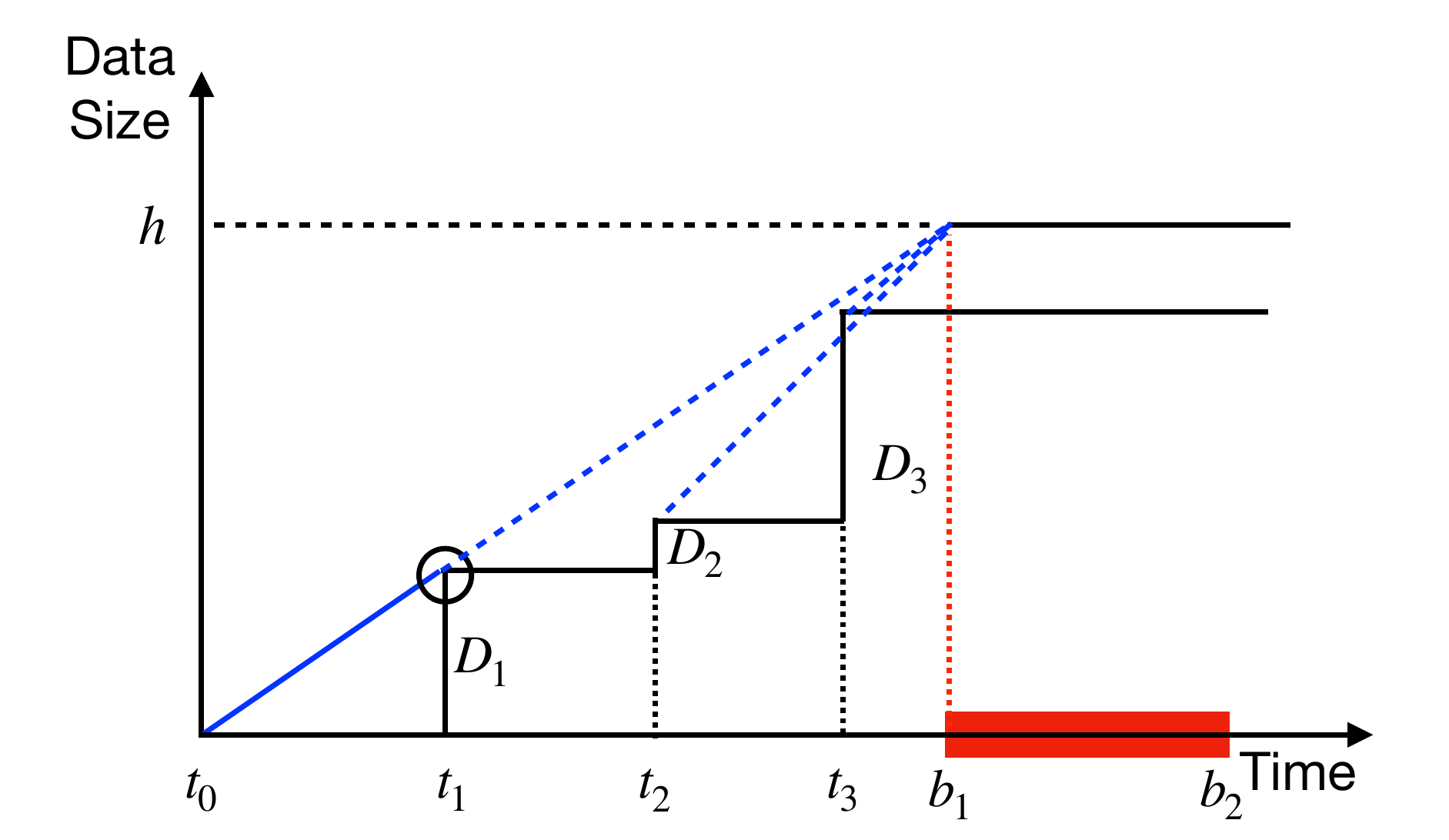}
\caption{Sub-areas determined by critical points}
\label{FigArea}
\end{figure}

\subsection{Searching Area Division}
After determining the searching area, the remaining work is to find the height that enables the optimal sensing and transmission rates. To execute the height search, the changing of total energy consumption with the varying height should be analyzed at first. According to Propositions~\ref{prop:LowBound} and \ref{prop:UpBound}, the sensing or transmission energy consumption will decrease or increase respectively as the height decreases in the searching area. As the total energy consumption is comprised of that both for sensing and transmission, whether it is increasing or decreasing with the change of height is hard to determine. Moreover, the SP structure will also change with the varying height. As shown in Fig.~\ref{FigArea}, the original rates before and after the point $(t_1,D_1)$ will no longer be equal if the height $h$ decreases. The points leading to such change are defined as the \emph{critical points}, which divide the the whole searching area into a series of sub-areas. Once the height drops from one sub-area to another, the expression of total energy consumption will change. Therefore, it is extremely complex to obtain the optimal height through traditional bisectional searching method if not impossible. 

\begin{figure}[t] 
\centering
\includegraphics[scale = 0.5]{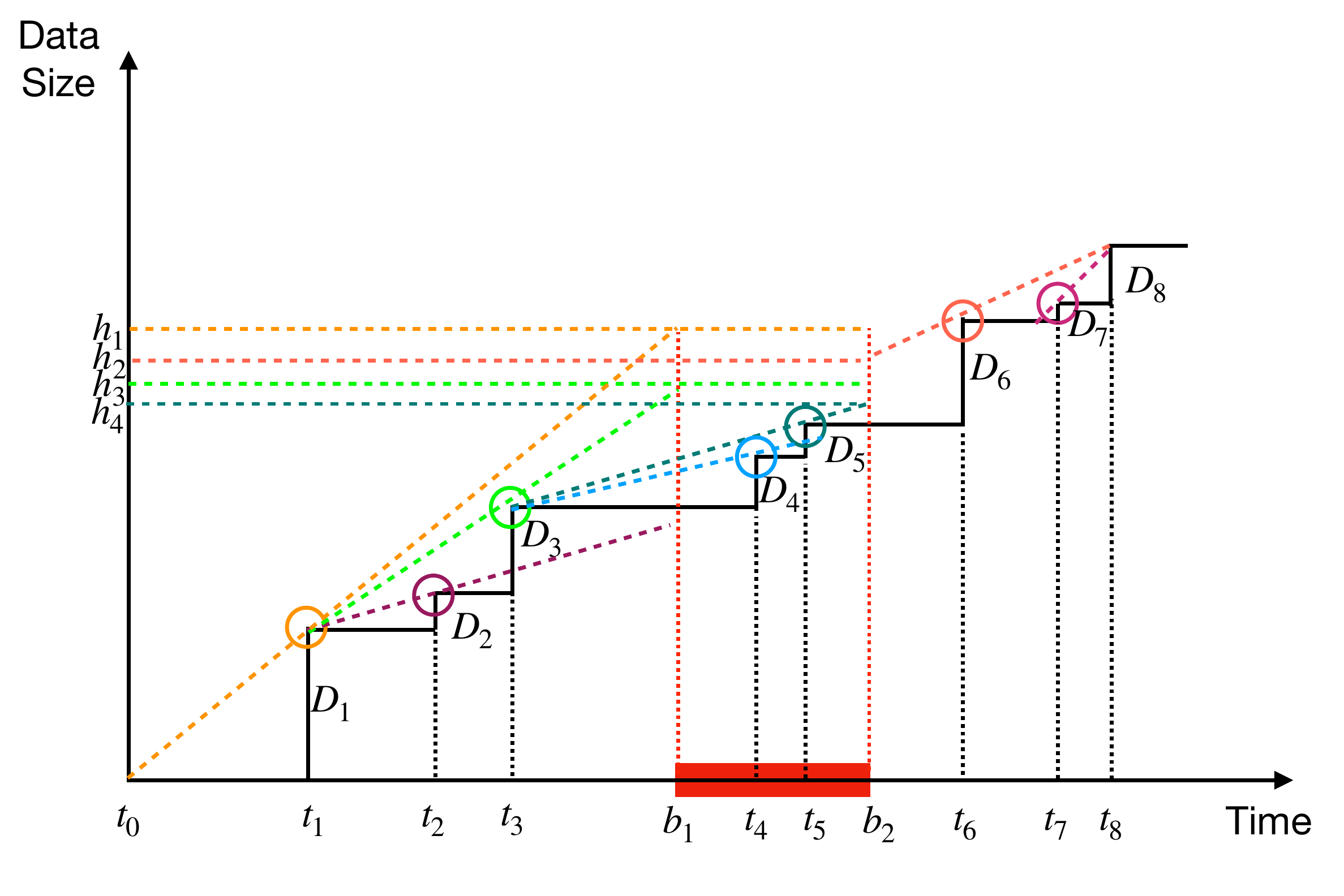}
\caption{Critical points and corresponding heights}
\label{FigCritical}
\end{figure}

To enable efficient searching, all the critical points and corresponding heights need to be found out. As shown in Fig.~\ref{FigCritical}, the critical points are the subset of $\{(t_j,\sum_{n=1}^j D_n)\}$ that might cause the change of SP structure denoted by $\{(t_m,\sum_{n=1}^m D_n)\}$. As for the case before the end of busy time, the point with the largest slope $\max_j \frac{\sum_{n=1}^{j} D_n}{t_j-t_0}$ should be the first critical point. By setting this critical point $(t_m,\sum_{n=1}^{m} D_n)$ as the new start, the next critical point is the point with the largest slope $\max_j \frac{\sum_{n=m+1}^{j} D_n}{t_j-t_m}$. As for the case after the end of busy time, the first critical point is the one with the smallest slope $\min_j \frac{\sum_{n=j+1}^{N} D_n}{t_N-t_j}$ denoted by $(t_m,\sum_{n=1}^{m} D_n)$, and the subsequent critical points can be found by checking the smallest slope $\min_j \frac{\sum_{n=j+1}^{m} D_n}{t_m-t_j}$. Therefore, all the critical points and corresponding heights can be found sequentially by the Algorithm~\ref{A4}. Specifically, as the height decreases, the first critical point is $(t_1,D_1)$ and the next three are $(t_6,\sum_{n=1}^6 D_n)$, $(t_3,\sum_{n=1}^3 D_n)$, and $(t_5,\sum_{n=1}^5 D_n)$. It should be noted that the points $(t_2,\sum_{n=1}^2 D_n)$, $(t_4,\sum_{n=1}^4 D_n)$ and $(t_7,\sum_{n=1}^7 D_n)$ are not critical since otherwise the data requirement constraints will be violated. Once a critical point $(t_m,\sum_{n=1}^{m} D_n)$ before the end of busy time is determined, the rate won't change with the height in duration $[t_0,t_m]$, which is known as the \emph{unchanged interval}. As for the critical point after the busy time interval, the unchanged interval is $[t_m,t_N]$. 

\begin{algorithm}[tt]
\renewcommand{\algorithmicrequire}{\textbf{Input:}}
\renewcommand{\algorithmicensure}{\textbf{Output:}}
\caption{Critical Points and Corresponding Heights Searching Algorithm}
\label{A4}
\begin{algorithmic}[1]
\REQUIRE Required data amounts $\{D_n\}$ at instants $\{t_n\}$ for $N$ tasks, busy time $[b_1,b_2]$
\ENSURE Critical points $\{(t_m,\sum_{n=1}^{m} D_n)\}$ in searching area and corresponding heights $\{h_m\}$.
\STATE Initialize $j_0 = 0$, $i = 0$.
\STATE \textbf{while} $t_{j_i} \leq b_1$
\STATE \qquad Update $i = i+1$.
\STATE \qquad Calculate $j_i = \arg\max_{j:t_{j_{i-1}}<t_j \leq b_1}\left\{\frac{\sum_{n=j_{i-1}+1}^{j}D_n}{t_j-t_{j_{i-1}}}\right\}$. 
\STATE \qquad Calculate $h_{j_i} = (b_1 - t_{j_{i-1}}) \frac{\sum_{n=j_{i-1}}^{j_i} D_n}{t_{j_i}-t_{j_{i-1}}} + \sum_{n=1}^{j_i} D_n$.
\STATE \textbf{end while}
\STATE Let $t_{j_i} = t_0$
\STATE \textbf{while} $t_{j_i} \leq b_2$
\STATE \qquad Update $i = i+1$.
\STATE \qquad Calculate $j_i = \arg\max_{j:t_{j_{i-1}}<t_j \leq b_2}\left\{\frac{\sum_{n=j_{i-1}+1}^{j}D_n}{t_j-t_{j_{i-1}}}\right\}$. 
\STATE \qquad Calculate $h_{j_i} = (b_2 - t_{j_{i-1}}) \frac{\sum_{n=j_{i-1}}^{j_i} D_n}{t_{j_i}-t_{j_{i-1}}} + \sum_{n=1}^{j_i} D_n$.
\STATE \textbf{end while}
\STATE Let $t_{j_i} = t_N$
\STATE \textbf{while} $t_{j_i} \leq t_N$
\STATE \qquad Update $i = i+1$.
\STATE \qquad Calculate $j_i = \arg\min_{j:b_2<t_j \leq t_{j_{i-1}}}\left\{\frac{\sum_{n=j+1}^{j_{i-1}}D_n}{t_{j_{i-1}}-t_j}\right\}$. \STATE \qquad Calculate $h_{j_i} = \sum_{n=1}^{j_i} D_n -  (t_{j_i} - b_2) \frac{\sum_{n=j_i+1}^{j_{i-1}}D_n}{t_{j_{i-1}}-t_{j_i}}$.
\STATE \textbf{end while}
\STATE Sort the heights in decreasing order denoted by $\{h_m\}$.
\STATE \textbf{Return} The critical points $\{(t_m,\sum_{n=1}^{m} D_n)\}$ and corresponding heights $\{h_m\}$.
\end{algorithmic}
\end{algorithm}


\subsection{Height searching}
After finding all the critical points and corresponding heights, the remaining steps are to determine the local optimal height that results in the smallest energy consumption in each sub-area, and get the global optimal height by comparing these local optimal heights. In each sub-area, the heights that might reach local optimum are given in the following proposition.

\begin{proposition}[Local optimal height in each sub-area]\label{prop:change}
\emph{The local optimal height in each sub-area locates either at the boundary of such sub-area or where the derivate of energy consumption with respect to the height equals to zero. }
\end{proposition}

\begin{proof}
According to the definition, the rates in unchanged interval won't change with the heights and so does the corresponding energy consumption. As for sensing, consider one critical point before $b_1$ denoted by $(t_{m_1},\sum_{n=1}^{m_1}D_n)$ and one after $b_2$ denoted by $(t_{m_2},\sum_{n=1}^{m_2}D_n)$. The changing part of sensing energy consumption w.r.t. the height $h$ in the sub-area determined by these points can be expressed as $E_s(h) = \alpha C^2 [(\frac{h-\sum_{n=1}^{m_1}D_n}{b_1-t_{m_1}})^2  (b_1-t_{m_1})+ (\frac{\sum_{n=1}^{m_2}D_n - h}{t_{m_2}-b_2})^2 (t_{m_2}-b_2)]$. As for transmission, the critical point before $b_2$ is denoted by $(t_{m_3},\sum_{n=1}^{m_3}D_n)$, while that after $b_2$ should be the same one as $(t_{m_2},\sum_{n=1}^{m_2}D_n)$ since the curves of transmission and sensing are overlapped according to Lemma~\ref{lemma:Overlap}. The changing part of transmission energy consumption w.r.t the height $h$ in this sub-area is $E_t(h) = \frac{\sigma^2}{g}[\exp(\frac{h-\sum_{n=1}^{m_3}D_n}{(b_2-t_{m_3})B}) (b_2-t_{m_3}) + \exp(\frac{\sum_{n=1}^{m_2}D_n - h}{(t_{m_2}-b_2)B}) (t_{m_2}-b_2) - 2]$. Therefore, the changing part of total energy consumption $E(h) =  E_s(h)+E_t(h) = \alpha C^2 [(\frac{h-\sum_{n=1}^{m_1}D_n}{b_1-t_{m_1}})^2 (b_1-t_{m_1}) + (\frac{\sum_{n=1}^{m_2}D_n - h}{t_{m_2}-b_2})^2 (t_{m_2}-b_2)] + \frac{\sigma^2}{g}[\exp(\frac{h-\sum_{n=1}^{m_3}D_n}{(b_2-t_{m_3})B}) (b_2-t_{m_3}) + \exp(\frac{\sum_{n=1}^{m_2}D_n - h}{(t_{m_2}-b_2)B}) (t_{m_2}-b_2) - 2]$. One can get $\frac{dE(h)}{dh} = 2 \alpha C^2 [\frac{h-\sum_{n=1}^{m_1}D_n}{b_1-t_{m_1}} - \frac{\sum_{n=1}^{m_2}D_n-h}{t_{m_2}-b_2}] + \frac{\sigma^2}{gB}[\exp(\frac{h-\sum_{n=1}^{m_3}D_n}{(b_2-t_{m_3})B})- \exp(\frac{\sum_{n=1}^{m_2}D_n - h}{(t_{m_2}-b_2)B})]$ and $\frac{d^2E(h)}{dh^2} = 2 \alpha C^2[\frac{1}{b_1-t_{m_1}} + \frac{1}{t_{m_2}-b_2}] + \frac{\sigma^2}{gB^2}[\exp(\frac{h-\sum_{n=1}^{m_3}D_n}{(b_2-t_{m_3})B})/(b_2-t_{m_3}) + \exp(\frac{\sum_{n=1}^{m_2}D_n - h}{(t_{m_2}-b_2)B})/(t_{m_2}-b_2)] > 0$. If the solution of $\frac{dE(h)}{dh} = 0$ lies in the sub-area, then it is the local optimal height. Otherwise the local optimal height lies in the boundary of this sub-area.
\end{proof}

\begin{algorithm}[tt]
\renewcommand{\algorithmicrequire}{\textbf{Input:}}
\renewcommand{\algorithmicensure}{\textbf{Output:}}
\caption{Global Optimal Rates Searching Algorithm}
\label{A5}
\begin{algorithmic}[1]
\REQUIRE Required data amounts $ \{D_n\} $ at instants $\{t_n\}$ for $N$ tasks, busy time $[b_1,b_2]$.
\ENSURE  The global optimal height $h^*$, sensing rates $s^*(t)$, and transmission rates $r^*(t)$. 
\STATE  Determine the searching area according to Propositions~\ref{prop:LowBound} and \ref{prop:UpBound}.
\STATE  Find all critical points $\{(t_m,\sum_{n=1}^m D_n)\}$ and corresponding heights $\{h_m\}$ by Algorithm~\ref{A4}.
\STATE  Obtain the local optimal height $\{h_m^*\}$ in each sub-area determined by the critical points according to Proposition~\ref{prop:change}.
\STATE  Get the global optimal height $h^* = \arg\min_{h_m^*} E(h_m^*)$.
\STATE Get the optimal sensing and transmission rates $s^*(t)$ and $r^*(t)$ by Algorithm~\ref{A3} based on $h^*$.
\STATE \textbf{Return} the global optimal height $h^*$, sensing rates $s^*(t)$, and transmission rates $r^*(t)$.
\end{algorithmic}
\end{algorithm}

Based on Proposition~\ref{prop:change}, all local optimal heights in each sub-area can be found. After comparing the corresponding total energy consumption of the local optimal heights from different sub-areas, one can get the global optimal height. The optimal sensing and transmission rates can then be obtained by Algorithm~\ref{A3}. The whole process is summarized in Algorithm \ref{A5}.

\section{Extension for finite data buffer capacity case}
Inspired by the policy in the previous section, the joint data sensing and transmission rates control for the case with finite data buffer capacity of the server is investigated in this section. Constrained by the maximum capacity $D_{max}$ of the buffer, the size of stored data at the server should obey the buffer constraints, i.e.,
\begin{equation} \label{TxBuf}
\text{(transmission data buffer)} \sum_{i=1}^j r_i T'_i  - \sum_{i=1}^{j-1} D'_i \leq D_{max}, j = 1,...N', 
\end{equation}
where $D_0 = 0$ represents that no data is required at instant $t_0$. After incorporating the data buffer constraints, the optimization problem can be formulated as
\begin{subequations}
\begin{align}
\min_{\{r_i \geq 0\},\{s_i \geq 0\}}  ~
&  \sum_{i=1}^{N'} \left [\alpha C^2 s_i^2+\frac{\sigma^2}{g}(e^{r_i/B}-1)\right]T'_i \label{Eq:P3a}\\
\text{(P3)} \qquad \text{s.t.} \qquad
& \sum_{i=1}^j s_i T'_i \geq \sum_{i=1}^j D'_i, j = 1,...N',\label{Eq:P3b}\\
& \sum_{i=1}^j r_i T'_i\geq \sum_{i=1}^j D'_i, j = 1,...N',\label{Eq:P3c}\\
& \sum_{i=1}^j r_i T'_i  - \sum_{i=1}^{j-1} D'_i \leq D_{max}, j = 1,...N',\label{Eq:P3d}\\
& \sum_{i=1}^j s_i T'_i \geq \sum_{i=1}^j r_i T'_i, j = 1,...N',\label{Eq:P3e}\\
& s_i = 0, \forall t'_i \in [b_1,b_2]. \label{Eq:P3f}
\end{align}
\end{subequations}
To solve problem (P3), the height search algorithm is modified as elaborated below.

\begin{algorithm}[tt]
\renewcommand{\algorithmicrequire}{\textbf{Input:}}
\renewcommand{\algorithmicensure}{\textbf{Output:}}
\caption{SP Algorithm for Optimal Rate Control with Finite Data Buffer.}
\label{A6}
\begin{algorithmic}[1]
\REQUIRE  required data amounts $\{D_n\}$ at instants $\{t_n\}$ for $N$ tasks, data buffer capacity $D_{\max}$. 
\ENSURE the optimal transmission rates $\{r_i^*\}$ and durations $\{T_i^*\}$.
\STATE Initialize $u_b = 0$, $u_1 = 0$, $i = 0$.
\STATE \textbf{while} $N > 0$
\STATE \qquad Update $i = i+1$.
\STATE \qquad \textbf{for} $j = 1,...,N$
\STATE \qquad \qquad $r_e[j] = \frac{\sum_{n=0}^j D_n}{t_j}$,
\STATE \qquad \qquad $r_f[j] = \frac{\sum_{n=0}^{j-1} D_n + D_{\max}}{t_j}$,
\STATE \qquad \qquad $\bold{r}[j] = [r_e[j],r_f[j]] = \{r|r_e[j] \leq r \leq r_f[j]\}$.
\STATE \qquad \textbf{end for}
\STATE \qquad Update $u_b = \max\left\{u|\bigcap_{j=1}^u\bold{r}[j]\neq \emptyset, j=1,2,...,N\right\}$.
\STATE \qquad \textbf{if} $u_b = N$ 
\STATE \qquad \qquad Update $u_1 = \max\left\{u|r_e[u] \in \bigcap_{j=1}^{u_b}\bold{r}[j]\right\}$, $r_i^* = r_e[u_1]$, $T_i^* = t_{u_1}$.
\STATE \qquad \textbf{else} 
\STATE \qquad \qquad \textbf{if} $\bold{r}[u_b+1]$ falls below $\bigcap_{j=1}^{u_b}\bold{r}[j]$
\STATE \qquad \qquad \qquad Update $u_1 = \max\left\{u|r_e[u] \in \bigcap_{j=1}^{u_b}\bold{r}[j]\right\}$, $r_i^* = r_e[u_1]$, $T_i^* = t_{u_1}$.
\STATE \qquad \qquad \textbf{else}
\STATE \qquad \qquad \qquad Update $u_1 = \max\left\{u|r_f[u] \in \bigcap_{j=1}^{u_b}\bold{r}[j]\right\}$, $r_i^* = r_f[u_1]$, $T_i^* = t_{u_1}$.
\STATE \qquad \qquad \textbf{end if}
\STATE \qquad \textbf{end if}
\STATE \qquad Update $N = N - u_1$, $t_n = t_{n+u_1} - t_{u_1}$, $D' = r_i^* T_i^* - \sum_{n=0}^{u_1} D_n$.
\STATE \qquad Update $D_n = D_{n+u_1}$, $D_1 = D_1 - D'$.
\STATE \textbf{end while}
\STATE \textbf{Return} the optimal transmission rates $\{r_i^*\}$ and durations $\{T_i^*\}$.
\end{algorithmic}
\end{algorithm}

\subsection{Searching Area Division for Determining Local Optimal Height}
Due to the finite data buffer capacity, the original upper and lower bounds derived in the previous section might not hold. Specifically, the upper bound is no longer determined by SP of transmission rate curve above the floor via Algorithm \ref{A1}, but should be determined by the one accounting for finite data buffer capacity in \cite{li2021data} as summarized in Algorithm \ref{A6}, which finds the feasible region determined by the data buffer capacity and the corresponding constant rates in a recursive manner.

As shown in Fig. \ref{FigLimUp}, the upper bound might be higher than the buffer capacity constraint. Moreover, due to Lemma \ref{Lemma:RateLimited}, the rate might increase when achieving the buffer capacity constraint. Therefore, the lower bound might be higher than the upper bound as the sensing rate before the busy time interval might be less than that after it. In this special case, the optimal height is given in the proposition below.

\begin{proposition}[Optimal height in Special Case]\label{prop:OptSpe}
\emph{If the lower bound is higher than the upper bound, the optimal height locates exactly at the lower bound. }
\end{proposition}

\begin{proof}
The minimum sensing energy consumption is achieved at the lower bound, while the transmission energy consumption remains unchanged when the height is moved from the upper bound to the lower bound. 
\end{proof}

\begin{figure}[t] 
\centering
\includegraphics[scale = 0.5]{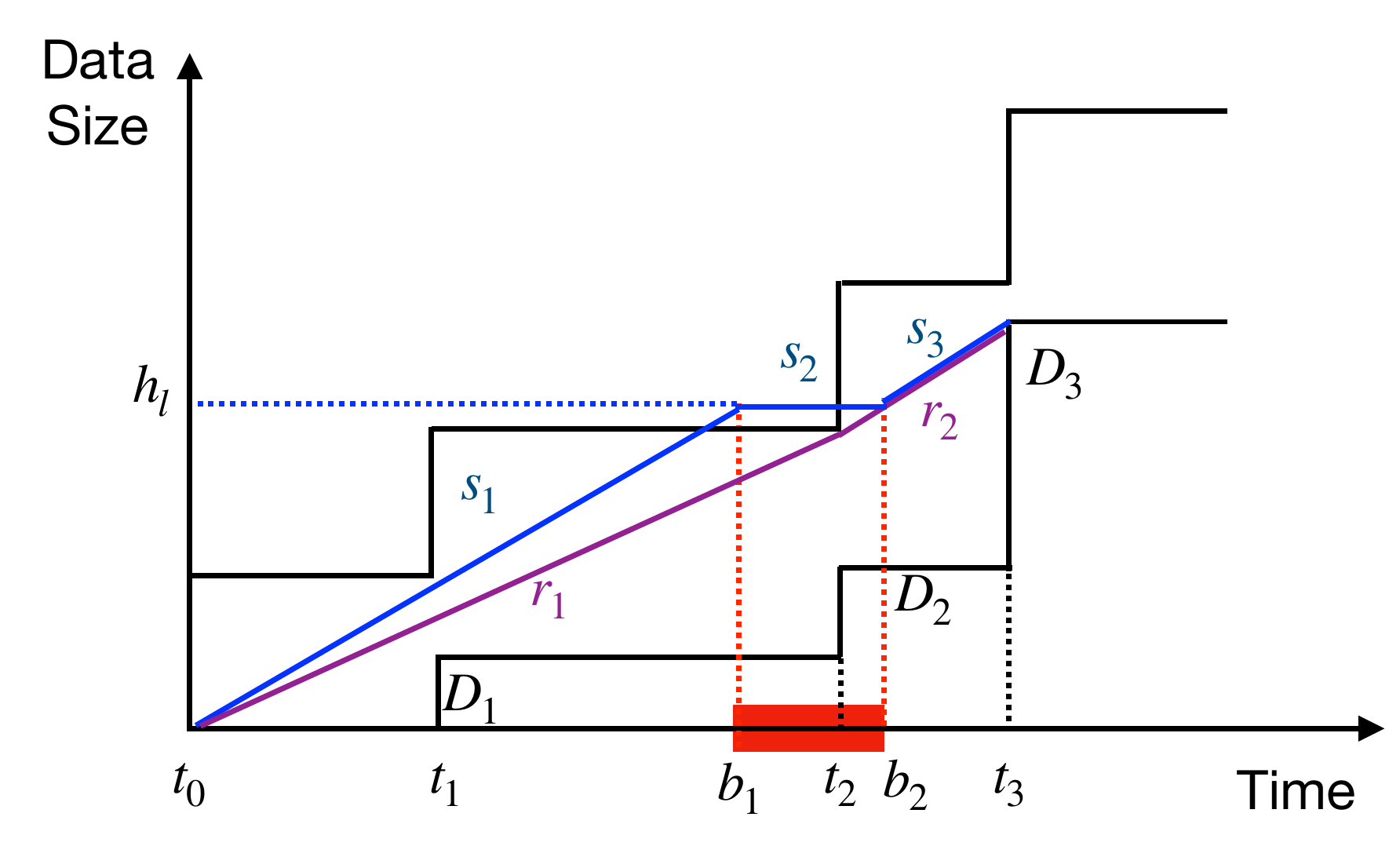}
\caption{Violated upper bound in the finite data buffer capacity case}
\label{FigLimUp}
\end{figure}

As shown in Fig. \ref{FigLimLow}, the lower bound derived in Proposition \ref{prop:LowBound} might also be higher than the buffer capacity constraint. In this case, if the upper bound is higher than the lower bound, then the optimal height still exists between the two bounds, while the transmission rate is constrained by the lower one of the optimal height or the buffer capacity constraint.

\begin{figure}[t] 
\centering
\includegraphics[scale = 0.5]{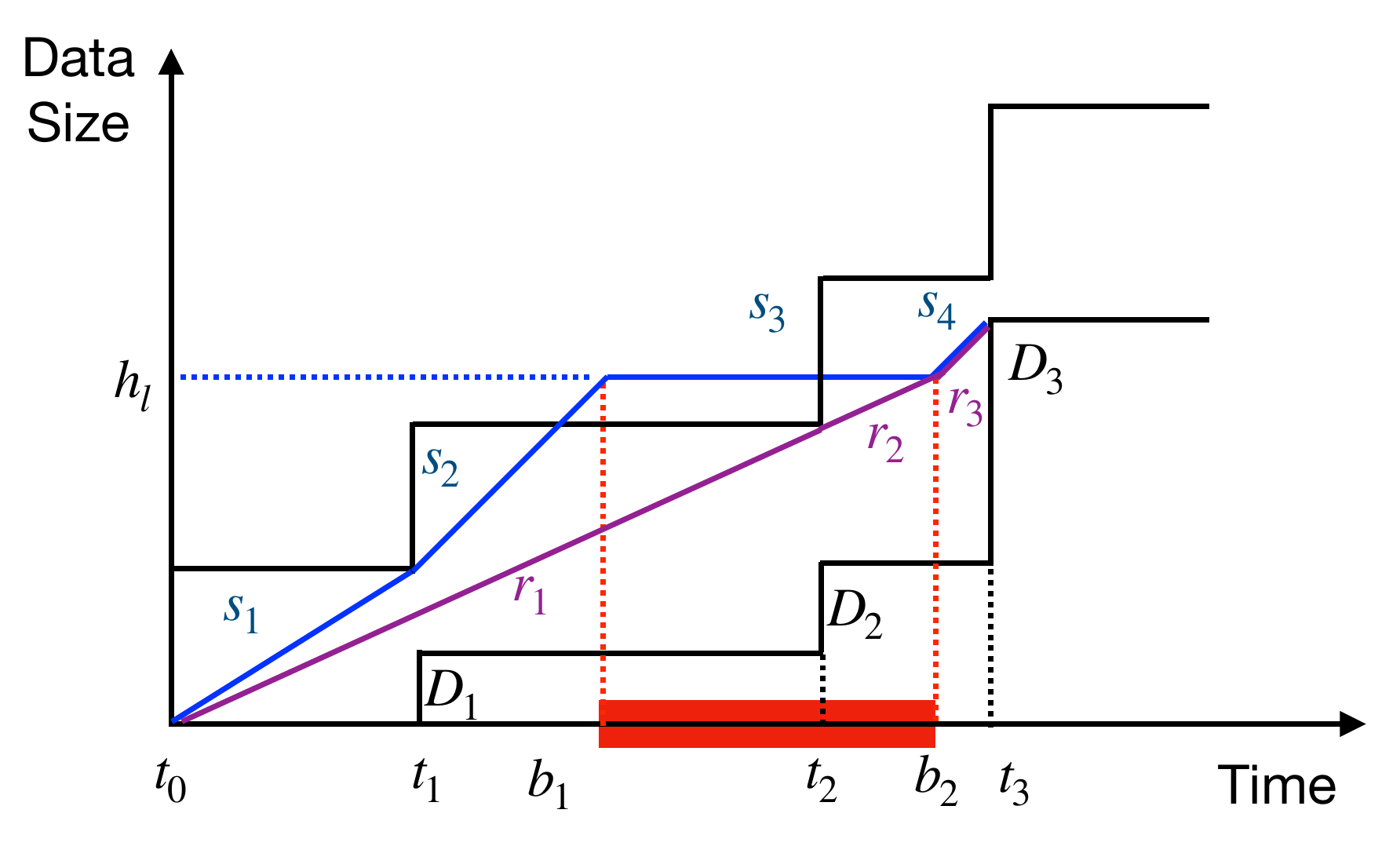}
\caption{Violated lower bound in the finite data buffer capacity case}
\label{FigLimLow}
\end{figure}

To execute the height searching, the changing of total energy consumption with the varying height should be analyzed at first. The data buffer constraints might cause the change of the SP structure with the varying height. For example, the sensing rates before and after the point $(t_3,D_3+D_{\max})$ become equal when the height decreases from $h_1$ to $h_2$ as shown in Fig. \ref{FigLimCritical}. The points leading to such changes are also known as the \emph{critical points}. Another type of critical point such as $(t_6,D_6)$ results in non-equal rates before and after it as defined in the previous section. The corresponding height at which the critical point is valid is defined as the \emph{critical height}. The critical heights divide the whole searching area into a series of sub-areas. 

\begin{figure}[t] 
\centering
\includegraphics[scale = 0.5]{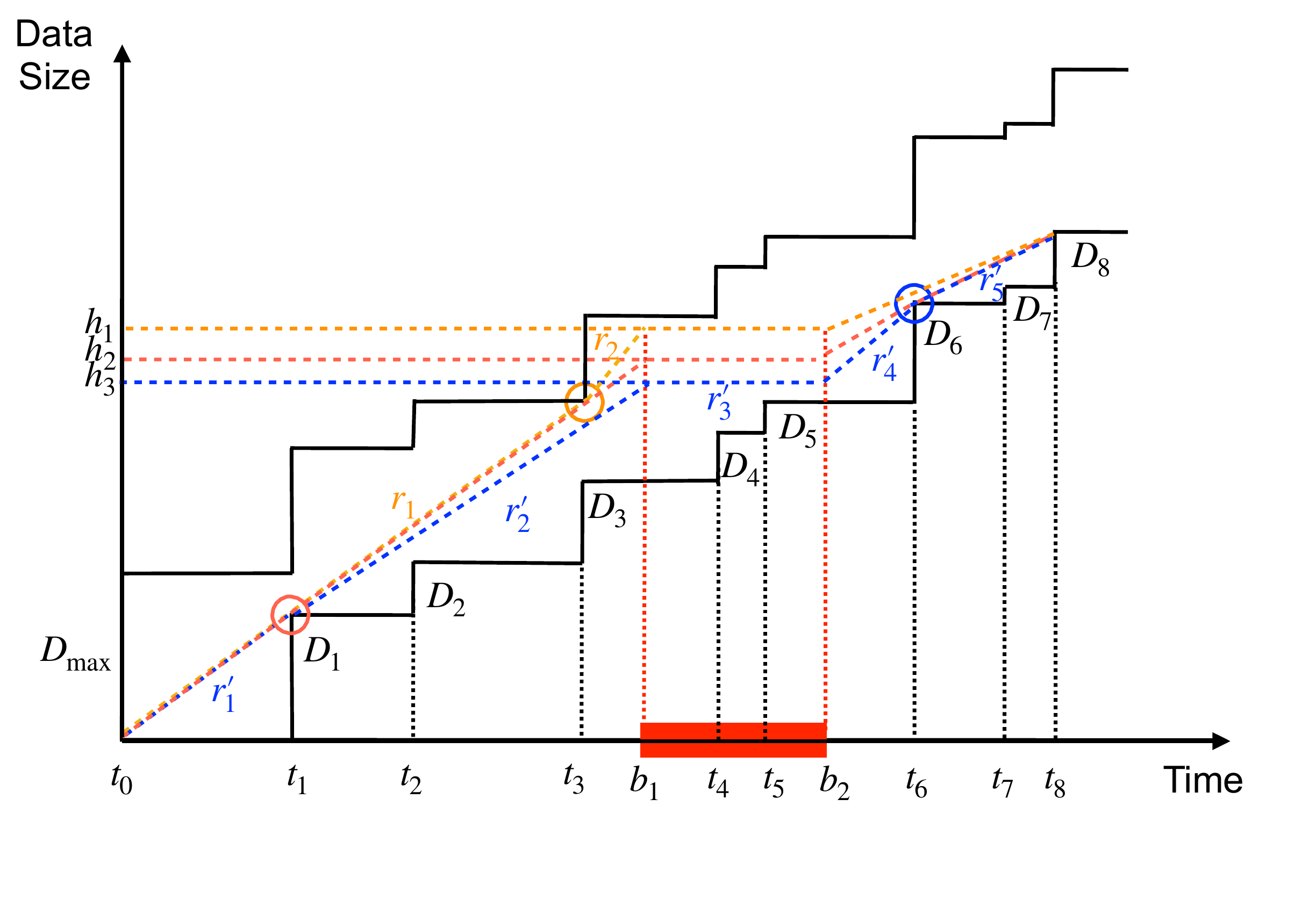}
\caption{Critical points and corresponding heights in the finite data buffer capacity case}
\label{FigLimCritical}
\end{figure}

The algorithm for finding the critical heights is summarized in Algorithm \ref{A7} and described below. The algorithm starts from a given height $h$, based on which the next critical height is found by comparing the slopes of the rate curves. A new variable namely busy time instant $b$ is introduced to represent one of $b_1$ and $b_2$. The slope of the non-zero rate curve that connects the point $(b,h)$ is defined as the nearest rate, while the non-zero rate curve that connects the nearest rate curve is defined as the second nearest rate curve. The time interval in which the rate remains as a constant is called the constant rate interval. For example, the nearest rate of point $(b_1,h_1)$ in Fig. \ref{FigLimCritical} is $r_2$ with the constant rate interval $[t_3,b_1]$, and the second nearest rate is $r_1$ with the constant rate interval $[t_0,t_3]$. Similarly, the nearest rate of point $(b_2,h_3)$ is $r_4'$ with the constant rate interval $[b_2,t_6]$, and the second nearest rate is $r_5'$ with the constant rate interval $[t_6,t_8]$. It should be noted that to find the next critical height, the Algorithm \ref{A7} need to be executed four times w.r.t. both the sensing and transmission rates before and after the busy time interval. The next critical height is the largest one among the four heights. All the critical heights are found by executing Algorithm \ref{A7} iteratively from the upper bound to the lower bound.


\begin{algorithm}[tt]
\renewcommand{\algorithmicrequire}{\textbf{Input:}}
\renewcommand{\algorithmicensure}{\textbf{Output:}}
\caption{Iterative Algorithm for Finding the Next Critical Height}
\label{A7}
\begin{algorithmic}[1]
\REQUIRE Required data amounts $ \{D_n\} $ at instants $\{t_n\}$ of $N$ tasks, busy time end $b$, height $h$, nearest rate $r_1$ and corresponding constant rate interval of $h: [t_{m_1},...,t_{n_1}]$, second-nearest rate $r_2$ and corresponding constant rate interval $h: [t_{m_2},...,t_{n_2}]$.
\ENSURE The next critical height $h$ in searching region and constant rate interval.
\STATE \textbf{if} $b \geq t_{n_1}$
\STATE \qquad Compute the slope $\frac{\sum_{n=m_1+1}^i D_n}{t_i-t_{m_1}}$ for every $t_i$ in $[t_{m_1 +1},...,t_{n_1}]$.
\STATE \qquad Find the point in $ [t_{m_1+1},...,t_{n_1}] $ with $\max_k{\frac{\sum_{n=m_1+1}^i D_n}{t_i-t_{m_1}}}$ denoted by $t_k$.
\STATE \qquad \textbf{if} $r_{1} < r_{2} $
\STATE \qquad \qquad Set $w = \frac{\sum_{m_1+1}^k D_n}{t_k - t_{m_1}}$.
\STATE \qquad \textbf{else if}  $r_{1} > r_{2} $
\STATE \qquad \qquad Set $w = \max\{\frac{\sum_{m_1+1}^k D_n}{t_k - t_{m_1}},r_{2}\}$. If $r_{2}$ is larger, set $k = m_2$.
\STATE \qquad \textbf{end if}
\STATE \qquad Get $h_{new} = w(b  - t_k)  + \sum_{n=0}^{k} D_n$, and the constant rate interval $[t_k,...,t_{n_1}] $.
\STATE \textbf{else if} $b \leq t_{m_1}$
\STATE \qquad Compute the slope $\frac{\sum_{n=i+1}^{n_1} D_n }{t_{n_1} - t_i}$ for every $t_i$ in $[t_{m_1},...,t_{n_1-1}]$.
\STATE \qquad Find the point in $ [t_{m_1},...,t_{n_1-1}] $ with $\min_k{\frac{\sum_{n=i+1}^n D_n }{t_{n_1} - t_i} }$ denoted by $ t_k $.
\STATE \qquad \textbf{if} $r_{1} > r_{2} $
\STATE \qquad \qquad Set $w = \frac{\sum_{k+1}^{n_1} D_n}{t_{n_1} - t_{k}}$.
\STATE \qquad \textbf{else if}  $r_{1} < r_{2} $
\STATE \qquad \qquad Set $w = \min\{\frac{\sum_{k+1}^{n_1} D_n}{t_{n_1} - t_{k}},r_{k+1}\}$. If $r_{2}$ is smaller, then set $k= n_2$.
\STATE \qquad \textbf{end if}
\STATE \qquad Get $ h_{new} = \sum_{n=0}^{k}D_n - w(t_k  - b)$, and the constant rate interval $[t_{m_1},...,t_{k}]$.
\STATE \textbf{end if}
\STATE \textbf{Return} The next critical height $h_{new}$, and the constant rate interval.
\end{algorithmic}
\end{algorithm}

As the expression of energy consumption does not change with the height in each sub-area, the local optimal height in each sub-area can be also obtained based on Proposition \ref{prop:change}.


\subsection{Optimal Rates based on Given Height}
Given a certain height, Algorithm \ref{A8} is proposed to determine the sensing and transmission rates. Specifically, the sensing rates before and after busy time interval are determined by Algorithm \ref{A1} as they are irrelevant with the data buffer constraints. Nevertheless, the transmission rates are constrained by both the finite data buffer capacity and the sensing rates. Therefore, the Algorithm \ref{A2} is applied to obtain the transmission rates with the maximum data size in each epoch settled as the lower one of the height and the data buffer constraint. Following the similar proof of Proposition \ref{prop:Optfixedh}, Algorithm \ref{A8} yields the optimal and feasible sensing and transmission rates given the fixed height $h$. The whole process for finding the optimal sensing and transmission rates is summarized in Algorithm \ref{A9}.

\begin{algorithm}[tt]
\renewcommand{\algorithmicrequire}{\textbf{Input:}}
\renewcommand{\algorithmicensure}{\textbf{Output:}}
\caption{SP Algorithm for Obtaining Optimal Rates with Finite Buffer given Fixed Height.}
\label{A8}
\begin{algorithmic}[1]
\REQUIRE Required data amounts $\{D_n\}$ at instants $\{t_n\}$ of  $N$ learning tasks, busy time interval $[b1,b2]$, height $h$, and data buffer capacity $D_{max}$
\ENSURE the optimal transmission rates $\{r_i^*\}$ and durations $\{{T_i^{r}}^*\}$, the optimal sensing rates $ \{s_i^*\} $ and durations $\{{T_i^{s}}^*\}$.
\STATE Initialize $n_1 = \max \{n|t_n < b_1\}$, $n_2 = \min\{n|t_n > b_2\}$, $n_3 = \max\{n|t_n < b_2\}$
\STATE Set $\{D_n^1\} = \{D_0,D_1,D_2,...,D_{n_1},h\}$ at instants $\{t_n^1\} = \{t_0,t_1,t_2,...,t_{n_1},b_1\}$
\STATE Set $\{D_n^2 \} = \{h,D_{n_2}-h,D_{n_2+1}-h,...,D_n-h\}$ at instants $\{t_n^2\} = \{b_2,t_{n_2},t_{n_2 + 1},...,t_n\}$
\STATE Given $\{D_n^1\},\{t_n^1\}$ and $\{D_n^2\},\{t_n^2\},D_{max}$, find the optimal sensing rates $\{s_i^1\}$, $\{s_i^2\}$ and durations $\{T_i^1\}$, $\{T_i^2\}$ by Algorithm \ref{A1}.
\STATE Get the optimal sensing rates $\{s_i ^*\} = \{\{s_i^1\},0,\{s_i^2\}\}$ and durations $\{{T_i^{s}}^*\} = \{\{T_n ^1\},b_2 - b_1,\{T_n ^2\}\}$.
\STATE Obtain $\{t_m\}$ by sorting the instants $\{\{t_n\},b_1,b_2\}$ in increasing order. 
\STATE Set $A(t_m) = \min\{h,\sum_{n=0}^{n_1}D_n + D_{max}\}$ if $t_m = b_1$, $A(t_m) = \min\{h,\sum_{n=0}^{n_3}D_n + D_{max}\}$ if $t_m = b_2 $. Otherwise, $ A(t_m) = \sum_{n=0}^{n-1}D_n + D_{max}$ if $t_m = t_n$.
\STATE Find the optimal transmission rates $\{r_i^*\}$ and durations $\{{T_i^{r}}^*\}$ by Algorithm \ref{A2}.
\STATE \textbf{Return} $\{r_i^*\}, \{T_i ^*\},\{s_i ^*\},\{T_i^{**}\}$
\end{algorithmic}
\end{algorithm}

\begin{algorithm}[tt]
\renewcommand{\algorithmicrequire}{\textbf{Input:}}
\renewcommand{\algorithmicensure}{\textbf{Output:}}
\caption{Global Optimal Rates Searching Algorithm in Finite Data Buffer Capacity Case}
\label{A9}
\begin{algorithmic}[1]
\REQUIRE Required data amounts $ \{D_n\} $ at instants $\{t_n\}$ of $N$ learning tasks, busy time $[b_1,b_2]$, data buffer capacity $D_{\max}$.
\ENSURE The global optimal height $h^*$, sensing rates $s^*(t)$, and transmission rates $r^*(t)$. 
\STATE Determine the lower and upper bounds $h_l$ and $h_u$. 
\STATE \textbf{if} $h_l \geq h_u$
\STATE \qquad Get the global optimal height $h^* = h_l$.
\STATE \textbf{else}
\STATE  \qquad Find all critical heights $\{h_m\}$ by executing Algorithm~\ref{A7} iteratively.
\STATE  \qquad Obtain the local optimal height $\{h_m^*\}$ in each sub-area according to Proposition~\ref{prop:change}.
\STATE \qquad  Get the global optimal height $h^* = \arg\min_{h_m^*} E(h_m^*)$.
\STATE \textbf{end if}
\STATE Get the optimal sensing and transmission rates $s^*(t)$ and $r^*(t)$ by Algorithm~\ref{A8}.
\STATE \textbf{Return} the global optimal height $h^*$, sensing rates $s^*(t)$, and transmission rates $r^*(t)$.
\end{algorithmic}
\end{algorithm}

\section{Simulation Results}
This section provides simulation results to evaluate the performance of the proposed algorithms. Each point in the figures is obtained by averaging over 100 simulation realizations, with independent channels in each realization. According to the settings in LTE \cite{schwarz2013lte}, the channels are under Rayleigh fading with bandwidth $B = 10$ MHz and noise power $\sigma^2 = -79.5$ dBm. Following the settings in \cite{you2018exploiting} and \cite{you2016energy}, the required number of CPU cycles per bit is $C = 500$ cycles/Bit, and the constant determined by the circuits is $\alpha = 10^{-28}$. Despite our proposed joint design
of joint sensing and transmission rate control denoted by JSTRC, three benchmark schemes are considered for performance comparison. The UB and LB schemes set the upper or lower bound on the searching area as the height and obtains the corresponding rates respectively, while the sensing and transmission rates are derived based on a random choosen height in the RH scheme.



\subsection{Joint Sensing and Transmission Rates Control with Infinite Data Buffer Capacity}
In the case with infinite data buffer capacity, there are $5$ tasks to be executed at time instants $\{10,20,80,90,200\}$ s with the required amount of data $\{500,500,500,700,300\}$ bits. The busy time interval is $[55,85]$ s. The specific sensing and transmission rates based on JSTRC scheme are depicted in Fig. \ref{FigUnSP}. One can observe that the SP structure holds for the optimal sensing and transmission rates control, which is in accordance with the theoretical analysis.

\begin{figure}[t] 
\centering
\includegraphics[scale = 0.5]{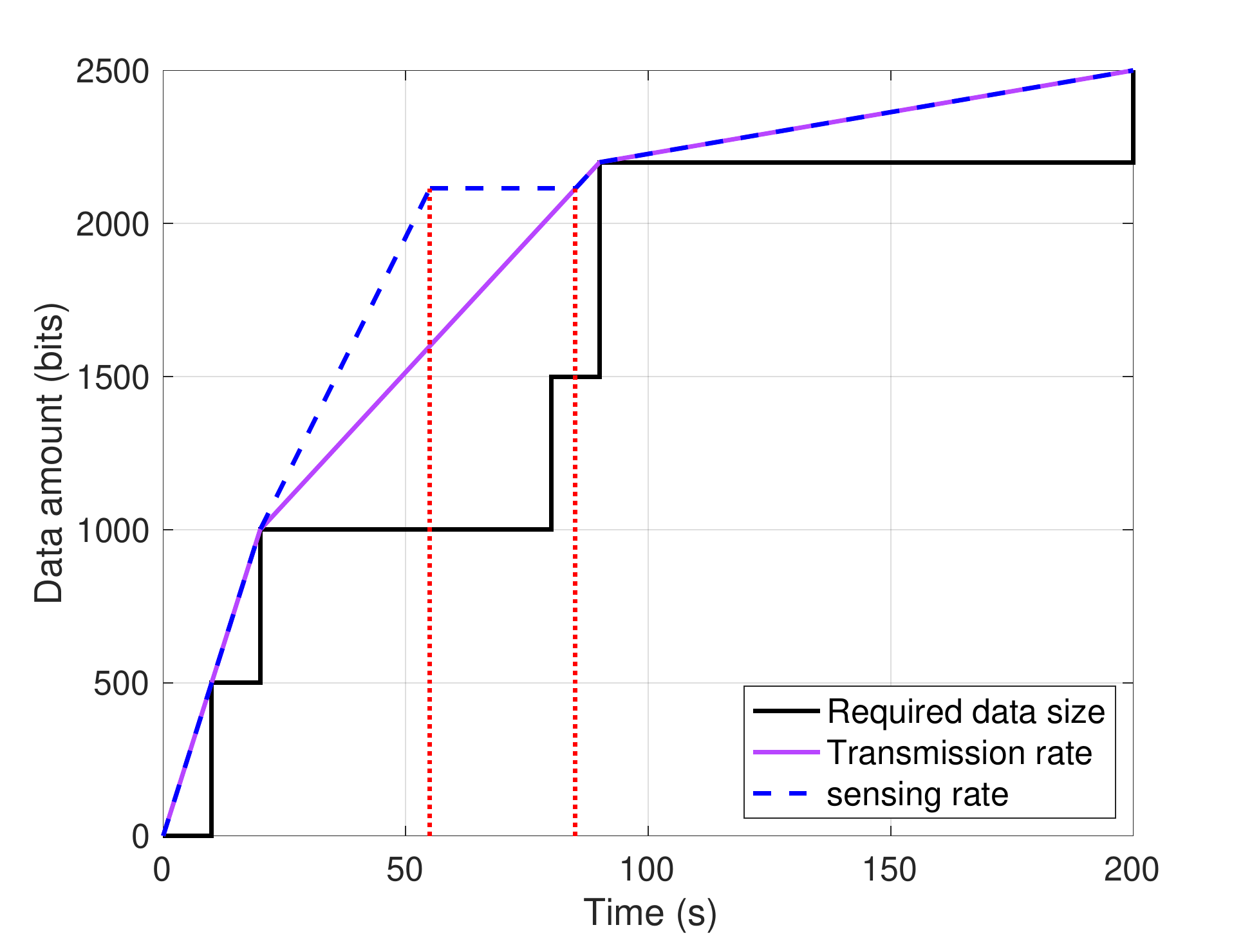}
\caption{Optimal data sensing and transmission rates in the infinite buffer capacity case}
\label{FigUnSP}
\end{figure}

The total energy consumption versus the total amount of required data by all tasks is illustrated in Fig. \ref{FigUnData}, where the amount of data required by each task increases proportionally. It can be observed that the total energy consumption will increase with the increasing amount of required data, as sensing and transmitting more data will result in higher energy consumption. Moreover, our proposed JSTRC scheme has the best performance comparing with other three baseline schemes, which verifies the effectiveness of height searching. It should also be noted that the LB scheme has the worst performance. The reason is that the optimal height is closer to the upper bound than the lower bound as the transmission process with exponential power-rate relationship dominates the total energy consumption.

\begin{figure}[t] 
\centering
\includegraphics[scale = 0.5]{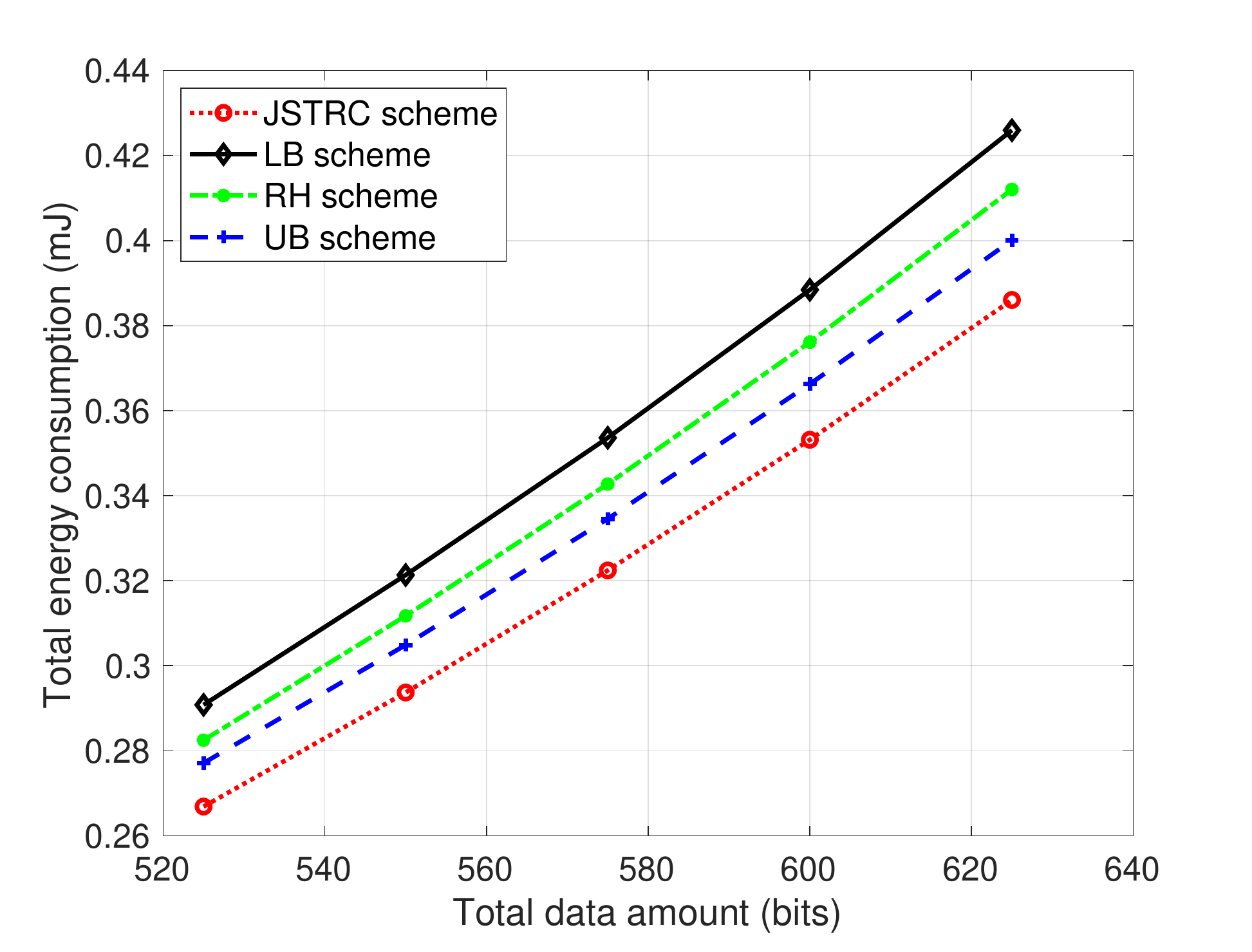}
\caption{Energy consumption versus the total data amount in the infinite buffer capacity case}
\label{FigUnData}
\end{figure}

Fig. \ref{FigUnTime} further demonstrates the total energy consumption versus the total time duration for all tasks, where the delay tolerance of each task increases proportionally. It can be observed that the total energy consumption decreases with the increasing time duration, as enlarging the duration can save the energy consumption for sensing and delivering the same amount of data. 

\begin{figure}[t] 
\centering
\includegraphics[scale = 0.5]{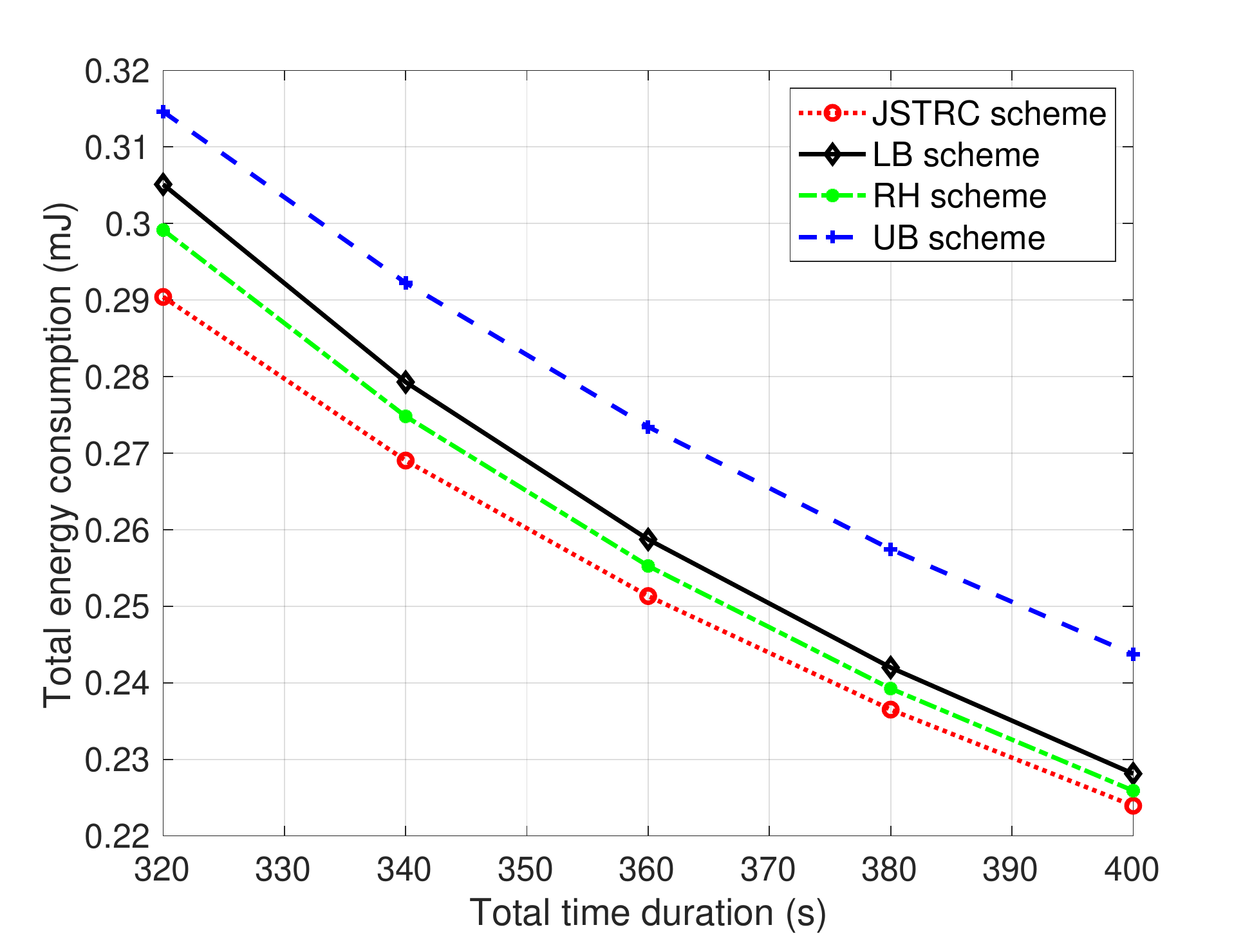}
\caption{Energy consumption versus the total time duration in the infinite buffer capacity case}
\label{FigUnTime}
\end{figure}



\subsection{Joint Sensing and Transmission Rates Control with Finite Data Buffer Capacity}
In the case with finite data buffer capacity, there are $5$ tasks to be executed at time instants $\{10,20,80,90,200\}$ s with the required amount of data $\{500,500,500,700,300\}$ bits. The data buffer capacity is $D_{\max} = 1000$ bits. The busy time interval is $[55,85]$ s. The specific sensing and transmission rates based on JSTRC scheme are depicted in Fig. \ref{FigLimitSP}. One can observe that the SP structure still holds for the optimal sensing and transmission rates control in the case with finite data buffer capacity.

\begin{figure}[t] 
\centering
\includegraphics[scale = 0.5]{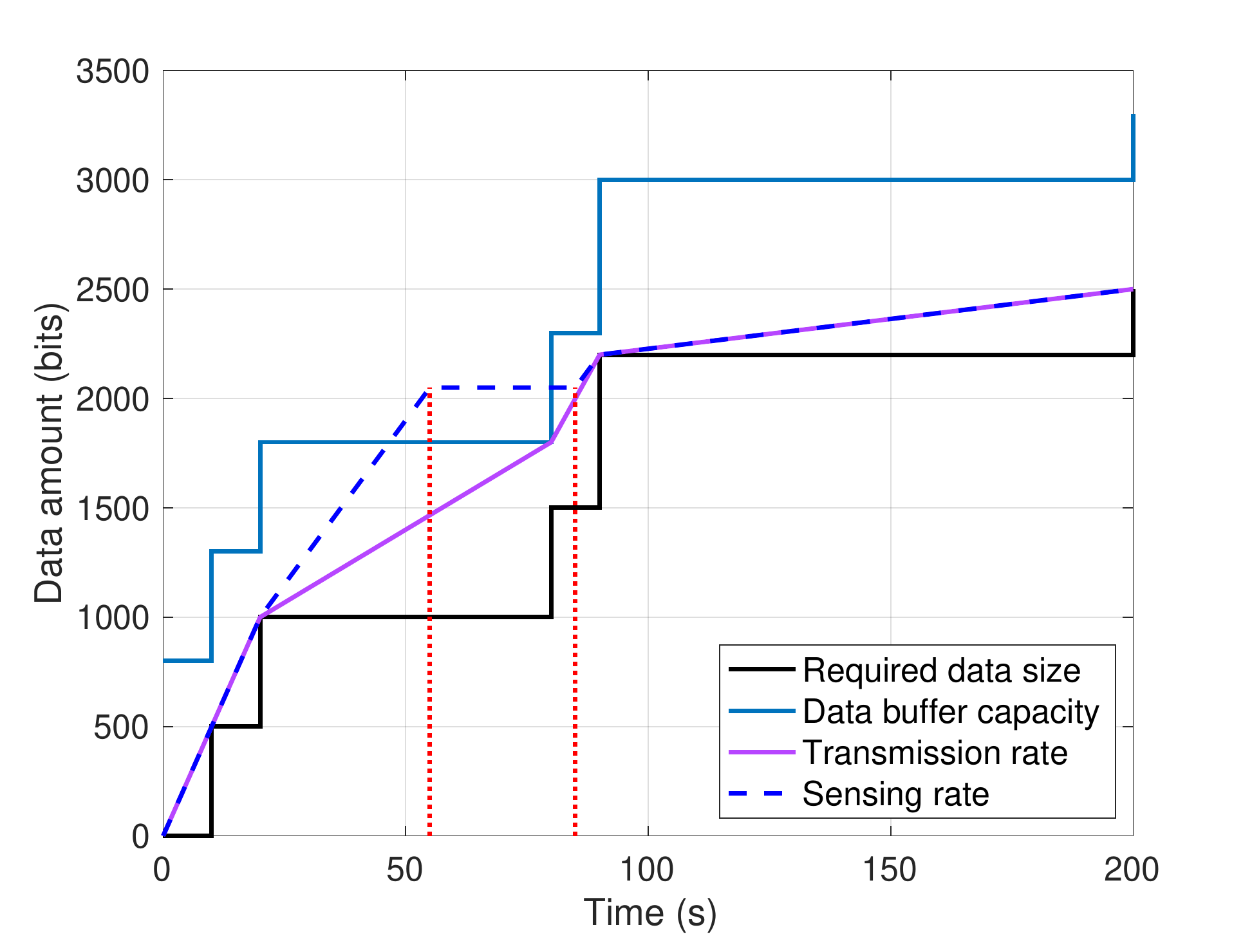}
\caption{Optimal data sensing and transmission rates in the finite buffer capacity case}
\label{FigLimitSP}
\end{figure}

The total energy consumption versus the total amount of required data in the case with finite data buffer capacity is illustrated in Fig. \ref{FigLimitData}, where the amount of data required by each task increases proportionally. It can be observed that our proposed JSTRC scheme still has the best performance. Moreover, the energy consumption in this case is larger than that with infinite data buffer capacity, since the finite data buffer capacity will add extra constraints on the original problem and makes the original optimum infeasible. 

\begin{figure}[t] 
\centering
\includegraphics[scale = 0.5]{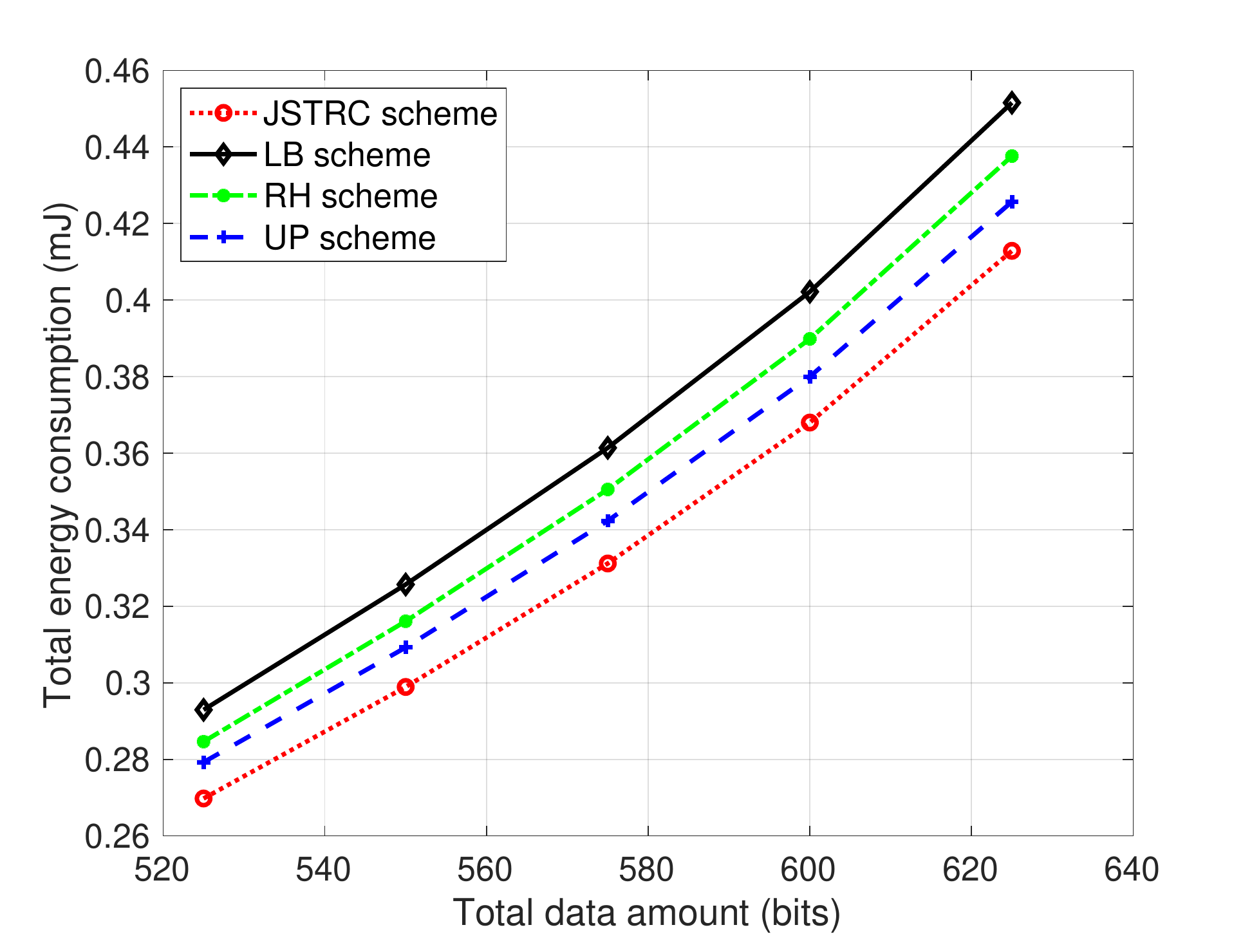}
\caption{Energy consumption versus the total data amount in the finite buffer capacity case}
\label{FigLimitData}
\end{figure}

Fig. \ref{FigLimitTime} demonstrates the total energy consumption versus the total time duration for all tasks in the case with finite data buffer capacity. It can be observed that our proposed JSTRC scheme still has the best performance, while the performance of LB scheme is close to optimal when the total time duration $T \geq 380$ s. The reason is that in such case the optimal height is close to the lower bound of the searching area. 

\begin{figure}[t] 
\centering
\includegraphics[scale = 0.5]{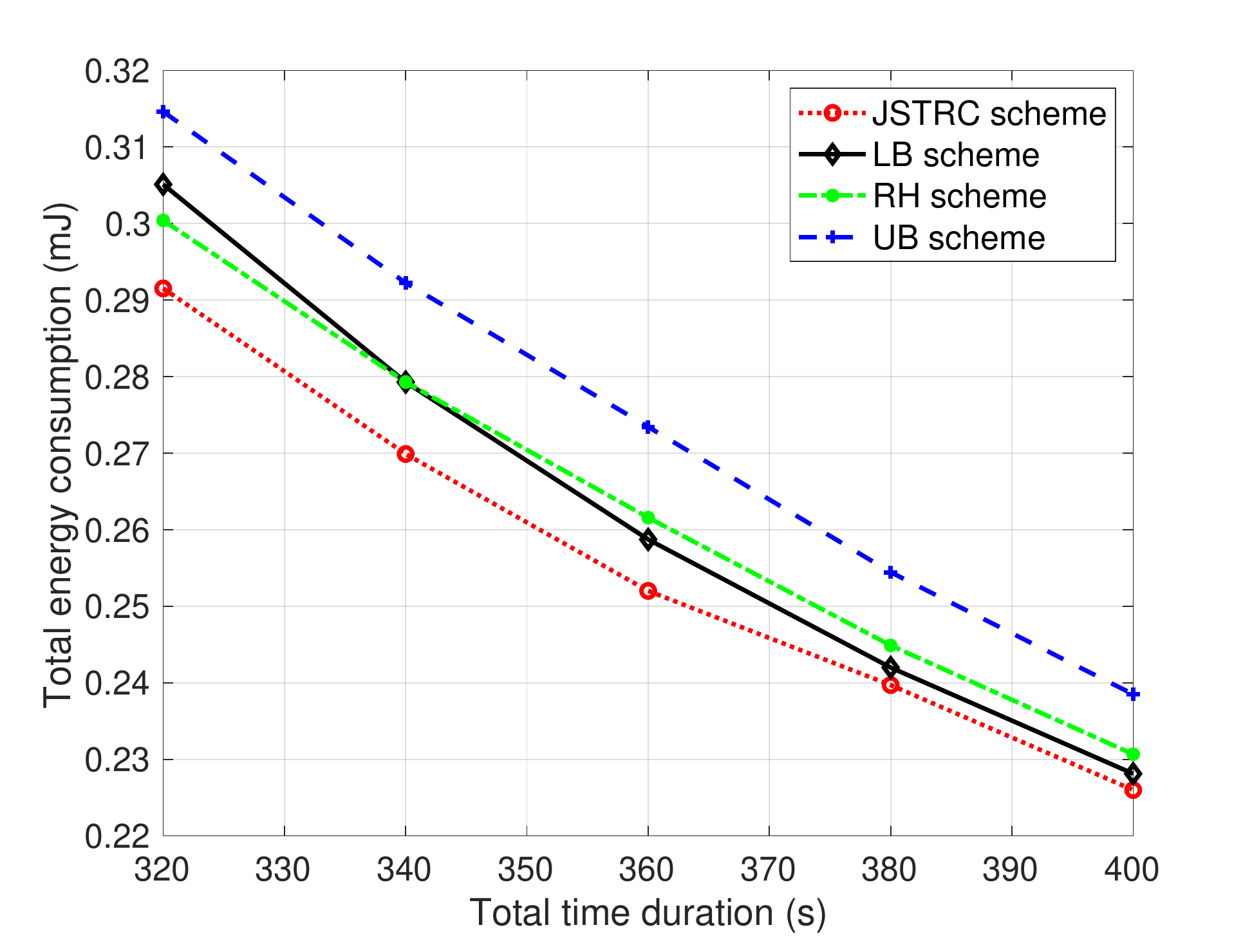}
\caption{Energy consumption versus the total time duration in the finite buffer capacity case}
\label{FigLimitTime}
\end{figure}

To show the effect of data buffer capacity on the performance, the total energy consumption versus the data buffer capacity is further depicted in Fig. \ref{FigBuffer}. It can be observed that our proposed JSTRC scheme has the minimum energy consumption, which is first decreasing with the increasing data buffer capacity and then saturates. The reason is that the data buffer constraints become less stringent as the data buffer capacity increases. Moreover, the energy consumption of the UB scheme first decreases and the increases with the increasing data buffer capacity. That is because the decreasing trend of transmission energy consumption is dominated at the beginning, while the increasing trend of sensing energy consumption becomes dominated as the data buffer capacity increases.

\begin{figure}[t] 
\centering
\includegraphics[scale = 0.5]{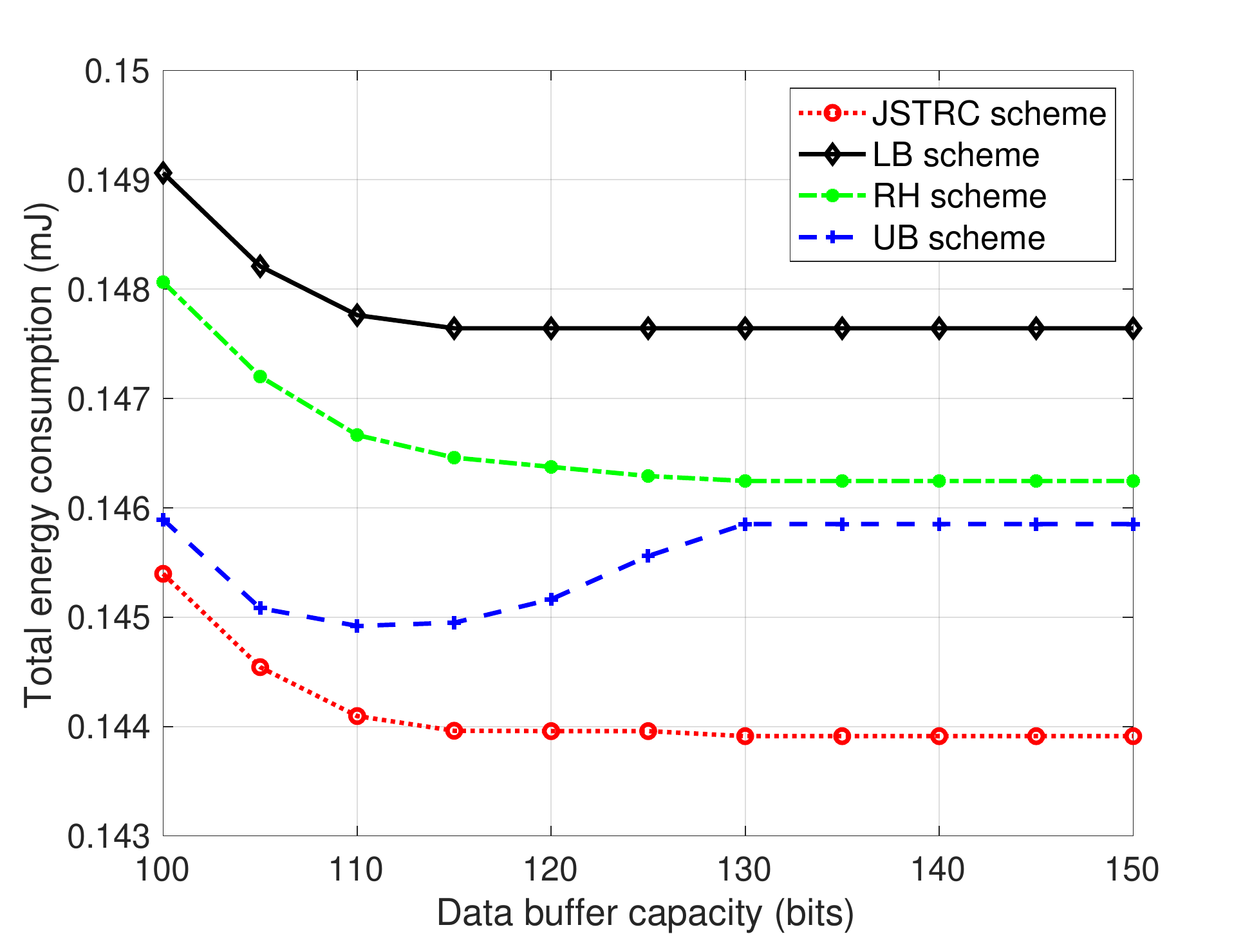}
\caption{Energy consumption versus the data buffer capacity}
\label{FigBuffer}
\end{figure}

\section{Conclusion}
In this paper, we investigate the design of joint sensing and transmission rates control for energy efficient mobile crowd sensing. The joint design is formulated as a complex optimization problem due to the intrinsic coupling between the controlling variables as well as the existence of busy time. To deal with such problem, the SP algorithms are exploited to determine the optimal sensing and transmission rates given the fixed height and an efficient algorithm is proposed to obtain the optimal height. Such a solution approach is further extended to account for the case with finite data buffer capacity at the server. This work opens a new direction for energy efficient joint sensing and transmission rates control. The performance of the schemes is verified via simulations. The current design can be extended for more complex scenarios with time-varying channels and multiple sensors.

\bibliographystyle{IEEEtran}

\begin{thebibliography}{}
\providecommand{\url}[1]{#1}
\csname url@samestyle\endcsname
\providecommand{\newblock}{\relax}
\providecommand{\bibinfo}[2]{#2}
\providecommand{\BIBentrySTDinterwordspacing}{\spaceskip=0pt\relax}
\providecommand{\BIBentryALTinterwordstretchfactor}{4}
\providecommand{\BIBentryALTinterwordspacing}{\spaceskip=\fontdimen2\font plus
\BIBentryALTinterwordstretchfactor\fontdimen3\font minus
  \fontdimen4\font\relax}
\providecommand{\BIBforeignlanguage}[2]{{%
\expandafter\ifx\csname l@#1\endcsname\relax
\typeout{** WARNING: IEEEtran.bst: No hyphenation pattern has been}%
\typeout{** loaded for the language `#1'. Using the pattern for}%
\typeout{** the default language instead.}%
\else
\language=\csname l@#1\endcsname
\fi
#2}}
\providecommand{\BIBdecl}{\relax}
\BIBdecl

\end{thebibliography}


\begin{thebibliography}{10}

\bibitem{chiang2016fog}
M.~Chiang and T.~Zhang, ``Fog and {IoT}: {An} overview of research
  opportunities,'' \emph{IEEE Internet Things J.}, vol.~3, no.~6, pp. 854--864,
  2016.

\bibitem{ma2014opportunities}
H.~Ma, D.~Zhao, and P.~Yuan, ``Opportunities in mobile crowd sensing,''
  \emph{IEEE Commun. Mag.}, vol.~52, no.~8, pp. 29--35, 2014.

\bibitem{akyildiz2002wireless}
I.~F. Akyildiz, W.~Su, Y.~Sankarasubramaniam, and E.~Cayirci, ``Wireless sensor
  networks: {A} survey,'' \emph{Comput. netw.}, vol.~38, no.~4, pp. 393--422,
  2002.

\bibitem{ganti2011mobile}
R.~K. Ganti, F.~Ye, and H.~Lei, ``Mobile crowdsensing: {Current} state and
  future challenges,'' \emph{IEEE Commun. Mag.}, vol.~49, no.~11, pp. 32--39,
  2011.

\bibitem{carpenter2002differentiated}
B.~E. Carpenter and K.~Nichols, ``Differentiated services in the {Internet},''
  \emph{Proc. IEEE}, vol.~90, no.~9, pp. 1479--1494, 2002.

\bibitem{zhang2014providing}
C.~Zhang, P.~Fan, K.~Xiong, and Y.~Dong, ``Providing differentiated services in
  multiaccess systems with and without queue state information,'' \emph{IEEE
  Trans. Commun.}, vol.~62, no.~12, pp. 4387--4400, 2014.

\bibitem{tao2008resource}
M.~Tao, Y.-C. Liang, and F.~Zhang, ``Resource allocation for delay
  differentiated traffic in multiuser {OFDM} systems,'' \emph{IEEE Trans.
  Wireless Commun.}, vol.~7, no.~6, pp. 2190--2201, 2008.

\bibitem{popovski20185g}
P.~Popovski, K.~F. Trillingsgaard, O.~Simeone, and G.~Durisi, ``{5G} wireless
  network slicing for {eMBB, URLLC, and mMTC: A} communication-theoretic
  view,'' \emph{IEEE Access}, vol.~6, pp. 55\,765--55\,779, 2018.

\bibitem{alsenwi2019embb}
M.~Alsenwi, N.~H. Tran, M.~Bennis, A.~K. Bairagi, and C.~S. Hong,
  ``{eMBB-URLLC} resource slicing: {A} risk-sensitive approach,'' \emph{IEEE
  Commun. Lett.}, vol.~23, no.~4, pp. 740--743, 2019.

\bibitem{li2018wirelessly}
X.~Li, C.~You, S.~Andreev, Y.~Gong, and K.~Huang, ``Wirelessly powered crowd
  sensing: {Joint} power transfer, sensing, compression, and transmission,''
  \emph{IEEE J. Sel. Areas Commun.}, vol.~37, no.~2, pp. 391--406, 2018.

\bibitem{prabhakar2001energy}
B.~Prabhakar, E.~U. Biyikoglu, and A.~El~Gamal, ``Energy-efficient transmission
  over a wireless link via lazy packet scheduling,'' in \emph{Proc. IEEE
  Infocom}, Anchorage, AK, USA, Aug. 2001.

\bibitem{zafer2005calculus}
M.~A. Zafer and E.~Modiano, ``A calculus approach to minimum energy
  transmission policies with quality of service guarantees,'' in \emph{Proc.
  IEEE Infocom}, Miami, FL, USA, Aug. 2005.

\bibitem{zafer2008optimal}
M.~A. Zafer and E.~Modiano, ``Optimal rate control for delay-constrained data transmission over a
  wireless channel,'' \emph{IEEE Trans. Inf. Theory}, vol.~54, no.~9, pp.
  4020--4039, 2008.

\bibitem{zafer2009calculus}
M.~A. Zafer and E.~Modiano, ``A calculus approach to energy-efficient data transmission with
  quality-of-service constraints,'' \emph{IEEE/ACM Trans. Netw.}, vol.~17,
  no.~3, pp. 898--911, 2009.

\bibitem{wang2013energy}
X.~Wang and Z.~Li, ``Energy-efficient transmissions of bursty data packets with
  strict deadlines over time-varying wireless channels,'' \emph{IEEE Trans.
  Wireless Commun.}, vol.~12, no.~5, pp. 2533--2543, 2013.

\bibitem{yang2011optimal}
J.~Yang and S.~Ulukus, ``Optimal packet scheduling in an energy harvesting
  communication system,'' \emph{IEEE Trans. Commun.}, vol.~60, no.~1, pp.
  220--230, 2011.

\bibitem{tutuncuoglu2012optimum}
K.~Tutuncuoglu and A.~Yener, ``Optimum transmission policies for battery
  limited energy harvesting nodes,'' \emph{IEEE Trans. Wireless Commun.},
  vol.~11, no.~3, pp. 1180--1189, 2012.

\bibitem{devillers2012general}
B.~Devillers and D.~G{\"u}nd{\"u}z, ``A general framework for the optimization
  of energy harvesting communication systems with battery imperfections,''
  \emph{J. Commun. Netw.}, vol.~14, no.~2, pp. 130--139, 2012.

\bibitem{gurakan2013energy}
B.~Gurakan, O.~Ozel, J.~Yang, and S.~Ulukus, ``Energy cooperation in energy
  harvesting communications,'' \emph{IEEE Trans. Commun.}, vol.~61, no.~12, pp.
  4884--4898, 2013.

\bibitem{ozel2013optimal}
O.~Ozel, J.~Yang, and S.~Ulukus, ``Optimal transmission schemes for parallel
  and fading {Gaussian} broadcast channels with an energy harvesting
  rechargeable transmitter,'' \emph{Computer commun.}, vol.~36, no.~12, pp.
  1360--1372, 2013.

\bibitem{wang2014optimal}
X.~Wang and R.~Zhang, ``Optimal transmission policies for energy harvesting
  node with non-ideal circuit power,'' in \emph{Proc. IEEE SECON}, Singapore,
  June 2014.

\bibitem{ulukus2015energy}
S.~Ulukus, A.~Yener, E.~Erkip, O.~Simeone, M.~Zorzi, P.~Grover, and K.~Huang,
  ``Energy harvesting wireless communications: A review of recent advances,''
  \emph{IEEE J. Sel. Areas Commun.}, vol.~33, no.~3, pp. 360--381, 2015.

\bibitem{luo2012training}
Y.~Luo, J.~Zhang, and K.~B. Letaief, ``Training optimization for energy
  harvesting communication systems,'' in \emph{Proc. IEEE GLOBECOM}, Anaheim,
  USA, Dec. 2012.

\bibitem{huang2012throughput}
C.~Huang, R.~Zhang, and S.~Cui, ``Throughput maximization for the {Gaussian}
  relay channel with energy harvesting constraints,'' \emph{IEEE J. Sel. Areas
  Commun.}, vol.~31, no.~8, pp. 1469--1479, 2012.

\bibitem{gregori2016wireless}
M.~Gregori, J.~G{\'o}mez-Vilardeb{\'o}, J.~Matamoros, and D.~G{\"u}nd{\"u}z,
  ``Wireless content caching for small cell and {D2D} networks,'' \emph{IEEE J.
  Sel. Areas Commun.}, vol.~34, no.~5, pp. 1222--1234, 2016.

\bibitem{you2018exploiting}
C.~You and K.~Huang, ``Exploiting non-causal {CPU}-state information for
  energy-efficient mobile cooperative computing,'' \emph{IEEE Trans. Wireless
  Commun.}, vol.~17, no.~6, pp. 4104--4117, 2018.

\bibitem{li2021data}
X.~Li, S.~Wang, G.~Zhu, Z.~Zhou, K.~Huang, and Y.~Gong, ``Data partition and
  rate control for learning and energy efficient edge intelligence,''
  \emph{arXiv preprint arXiv:2107.08884}, 2021.

\bibitem{lane2010survey}
N.~D. Lane, E.~Miluzzo, H.~Lu, D.~Peebles, T.~Choudhury, and A.~T. Campbell,
  ``A survey of mobile phone sensing,'' \emph{IEEE Commun. Mag.}, vol.~48,
  no.~9, pp. 140--150, 2010.

\bibitem{khan2012mobile}
W.~Z. Khan, Y.~Xiang, M.~Y. Aalsalem, and Q.~Arshad, ``Mobile phone sensing
  systems: {A} survey,'' \emph{IEEE Commun. Surveys Tuts.}, vol.~15, no.~1, pp.
  402--427, 2012.

\bibitem{chandrakasan1992low}
A.~P. Chandrakasan, S.~Sheng, and R.~W. Brodersen, ``Low-power {CMOS} digital
  design,'' \emph{IEICE Trans. Electron.}, vol.~75, no.~4, pp. 371--382, 1992.

\bibitem{you2016energy}
C.~You, K.~Huang, and H.~Chae, ``Energy efficient mobile cloud computing
  powered by wireless energy transfer,'' \emph{IEEE J. Sel. Areas Commun.},
  vol.~34, no.~5, pp. 1757--1771, 2016.

\bibitem{goldsmith2005wireless}
A.~Goldsmith, \emph{Wireless communications}. Cambridge university press, 2005.

\bibitem{schwarz2013lte}
O.~W. Schwarz and R.~Minihold, ``{LTE} system specifications and their impact
  on {RF} \& base band circuits,'' \emph{Rohde \& Schwarz App Note}, pp. 1 --
  37, 2013.
  
\end{thebibliography}

\end{document}